\documentclass[aps,prx,reprint]{revtex4-1}
\usepackage{blindtext}
\usepackage{amsmath}
\usepackage{amsfonts}
\usepackage{amsthm}
\usepackage{physics}
\usepackage{comment}
\usepackage{makecell}
\usepackage[
    colorlinks=true,
    hypertexnames=false
    ]{hyperref}

\hypersetup{
  colorlinks   = true,
  urlcolor = magenta!90!black,
  linkcolor = blue!60!black,
  citecolor=black!60
}

\usepackage[all]{hypcap}

\usepackage[capitalise,compress]{cleveref}
\usepackage{enumitem}   
\usepackage{tikz} 
\usetikzlibrary{arrows, shapes.gates.logic.US, calc}
\usetikzlibrary{shapes}
\usetikzlibrary{plotmarks}
\usetikzlibrary{quantikz}
\usetikzlibrary{patterns}

\newcommand{\imag}{\mathrm{i}}

\DeclareMathOperator{\TVD}{TVD}

\DeclareMathOperator{\poly}{poly}

\usepackage{xcolor}
\usepackage{bbm}

\newtheorem{theorem}{Theorem}
\newtheorem{lemma}{Lemma}
\newtheorem{definition}{Definition}

\newtheorem{conjecture}{Conjecture}

\crefname{protocol}{Protocol}{Protocols}
\Crefname{protocol}{Protocol}{Protocols}

\newlength{\protowidth}
\newcommand{\pprotocol}[5]{
{\begin{figure*}[#4]
\renewcommand\figurename{Protocol}
\renewcommand{\thefigure}{1}
\begin{center}
\setlength{\protowidth}{\textwidth}
        {
        \hrulefill \vspace{5pt}
        \small
        {\quad
        \begin{minipage}{\protowidth}
        \begin{center}
        {\bf #1}
        \end{center}
        #5

        \hrulefill

        \end{minipage}
        \quad}
        }

        \caption{\label[protocol]{#3} #2}
\end{center}
\vspace{-4ex}
\end{figure*}
} }

\newcommand{\protocol}[4]{
\pprotocol{#1}{#2}{#3}{tbh!}{#4} }

\begin{document}
\title{Efficiently verifiable quantum advantage on near-term analog quantum simulators}
\author{Zhenning~Liu$^{1,2,3}$, Dhruv~Devulapalli$^{1,2}$, Dominik~Hangleiter$^{1}$, Yi-Kai~Liu$^{1,4}$, Alicia~J.~Kollár$^{2,5}$, Alexey~V.~Gorshkov$^{1,2}$, and Andrew~M.~Childs$^{1,3,6}$  \\
\footnotesize $^{1}$ Joint Center for Quantum Information and Computer Science, NIST/University of Maryland, College Park,
Maryland 20742, USA\\
\footnotesize $^{2}$ Joint Quantum Institute, NIST/University of Maryland, College Park, Maryland 20742, USA\\
\footnotesize $^{3}$ Department of Computer Science, University of Maryland, College Park, Maryland 20742, USA\\
\footnotesize $^{4}$ Applied and Computational Mathematics Division,
National Institute of Standards and Technology (NIST), Gaithersburg, Maryland 20899, USA\\
\footnotesize $^{5}$ Department of Physics, University of Maryland, College Park, Maryland 20742, USA\\
\footnotesize $^{6}$ Institute for Advanced Computer Studies, University of Maryland, College Park, Maryland 20742, USA\\
}

\begin{abstract}
Existing schemes for demonstrating quantum computational advantage are subject to various practical restrictions, including the hardness of verification and challenges in experimental implementation. Meanwhile, analog quantum simulators have been realized in many experiments to study novel physics. In this work, we propose a quantum advantage protocol based on \emph{single-step Feynman-Kitaev} verification of an analog quantum simulation, in which the verifier need only run an $O(\lambda^2)$-time classical computation, and the prover need only prepare $O(1)$ samples of a history state and perform $O(\lambda^2)$ single-qubit measurements, for a security parameter $\lambda$. We also propose a near-term feasible strategy for honest provers and discuss potential experimental realizations.
\end{abstract}

\maketitle

\section{Introduction}
\subsection{Background \& Motivation}
Quantum computers offer the promise of executing some computational tasks exponentially faster than classical computers. This suggests a violation of the extended Church-Turing thesis, which says that any physically realizable model of computation can be efficiently simulated by a classical Turing machine. Indeed, quantum computers were originally proposed as a means of simulating quantum mechanical systems \cite{Feynman_1982}, a task considered classically hard. There has been much progress toward identifying classically difficult problems that quantum computers can solve efficiently, such as integer factorization \cite{shor1997factoring}, simulating Hamiltonian dynamics \cite{Lloyd_1996, childs2018toward,berry2007efficient}, and extracting information about 
solutions of high-dimensional linear systems \cite{harrow2009quantum}.

A significant milestone for the field of quantum computing is the first demonstration that a quantum device can perform computational tasks that a classical device with comparable resources cannot. This milestone has been called quantum supremacy \cite{arute2019quantum, preskill2012quantum}, quantum advantage, or a proof of quantumness \cite{brakerski2021cryptographic}, and has instigated numerous theoretical proposals and experimental efforts. However, there remain formidable technological challenges to building quantum computers, requiring both theoretical and experimental progress in architecture design, fault tolerance, and control. Various proposals for quantum advantage have addressed these challenges in different ways, by making trade-offs between ease of experimental demonstration, ease of verification, security guarantees, and practical applications.

\emph{Analog} quantum simulation \cite{Georgescu_Ashhab_Nori_2014}, i.e., using one many-body quantum system to simulate another, is a natural approach to demonstrating quantum advantage. By building quantum systems with tunable (but perhaps non-universal) Hamiltonians, one can emulate a large class of Hamiltonians that may be difficult to simulate classically. Since it directly encodes hard problems into controllable quantum systems, analog simulation arguably mitigates many of the issues faced by digital approaches \cite{blatt2012quantum,flannigan2022propagation}. Furthermore, analog simulation avoids Trotter error and other sources of algorithmic error in digital quantum simulation \cite{childs2021theory, cubitt2018universal}. Indeed, analog simulations of systems with hundreds of qubits have already been performed \cite{ebadi2021quantum}.

A major challenge for both quantum simulation and more general forms of quantum computation is the difficulty of verifying the correctness of a quantum process.
There have been several proposals to verify digital quantum computation \cite{fitzsimons2018post, mahadev2018classical} based on the \emph{Feynman-Kitaev circuit-to-Hamiltonian mapping} \cite{kitaev2002classical}, but such protocols are neither designed for analog quantum simulation nor practical on near-term analog quantum devices. Previous work towards verifying analog simulation has suggested approaches such as cross-platform verification \cite{elben2020crossplatform, carrasco2021theoretical}, Hamiltonian learning \cite{carrasco2021theoretical}, and performing a Loschmidt echo \cite{cirac2012goals, carrasco2021theoretical, shaffer2021practical}. Unlike protocols for digital verification, these approaches can be spoofed by dishonest or inaccurate quantum simulators, and therefore cannot be used to demonstrate quantum advantage in a sound, efficiently verifiable way. A step toward verified analog simulation is made in  \cite{hangleiter_direct_2017}, in which the verifier measures the energy of a \emph{parent Hamiltonian} of the output state of analog quantum simulation.
However, all these works require a significant number of samples of the simulator's state to certify it.

\subsection{Our Contribution}
In this paper, by combining a \emph{single-step} Feynman-Kitaev encoding and the scheme of \textcite{bermejo2018architectures}, we propose a novel quantum advantage protocol with reduced resource requirements, where a verifier capable of polynomial-time classical computation can verify the result by asking the prover to perform \emph{trusted measurements} on a \emph{constant} number of copies of a state. We also present a strategy for the honest prover and argue that it is feasible on near-term devices.

\textbf{The Protocol.} Our protocol involves interaction between a polynomial-time classical \emph{verifier}, and a quantum \emph{prover} who can do polynomial-time quantum computation, although we present a strategy for an honest prover who must only perform analog quantum simulation and some limited additional operations. In our protocol, the prover is capable of single-qubit \emph{trusted measurements}, which means that he performs the correct single-qubit measurements as instructed by the verifier with error rate $\epsilon = O(1/n$) (with $n$ the number of qubits), and reports the outcome honestly. We also allow a polynomial amount of classical communication in both directions.

Our protocol still works without the assumption of trusted measurements if the prover can send polynomial-size quantum states to the verifier, and the verifier can perform single-qubit measurements, as in the notion of a non-interactive $\mathsf{QPIP}_1$ protocol defined by \textcite{aharonov2017interactive} (where $\mathsf{QPIP}$ stands for \emph{quantum prover interactive proof}).

\begin{definition}[$\mathsf{QPIP}_k$ protocol (simplified)]
An interactive proof for a language $\mathcal{L}$ is said to be $\mathsf{QPIP}_k$ if the prover is a $\mathsf{BQP}$ machine, the verifier is a hybrid $\mathsf{BQP}$-$\mathsf{BPP}$ machine that can process at most $k$ qubits at a time, and quantum states of $k$ qubits can be transmitted from the prover to the verifier.
\end{definition}
However, as reliably sending quantum states is unlikely to be feasible in the near term, we focus on the former model.

\textbf{Prover's Model of Computation.} We also give an experimentally practical strategy for honest provers. The strategy is specifically designed for near-term machines that are not capable of fully digital quantum computation, but are slightly more powerful than \emph{analog} quantum simulation, a popular notion that is often not clearly defined. In our work, we define a \emph{mostly analog} model of computation, its \emph{commuting} version, and its extension with a global $CZ$ gate, which we argue are feasible for near-term experiments.
\begin{definition}[Mostly analog quantum computation]
\label{defn:mostly}
    A model of quantum computation involving $n$ qubits is called \emph{mostly analog} if all the following conditions hold. (1) The system can evolve under a \emph{time-independent} 2-body Hamiltonian $H$ containing $\poly(n)$ Pauli terms for time $T=\poly(n)$. 
    (2) $O(1)$ alternations between the evolution under $H$ and single-qubit rotations can be performed.
    (3) Measurements can only be performed once at the end of the whole process.
\end{definition}

Note that condition (2) distinguishes this model from common notions of analog quantum computation, as it may require a degree of control not always available to analog quantum simulators.
Despite being mostly analog, the above model of computation is even capable of $\mathsf{BQP}$-complete quantum computations \cite{nagaj2012universal}. We introduce a weaker model where the 2-local Hamiltonian $H$ must also be \emph{commuting}, which means that all Pauli terms must commute with each other.
\begin{definition}[Mostly analog commuting quantum computation]
    A mostly analog model of computation is called \emph{commuting} if $H$ is a commuting Hamiltonian.
\end{definition}

Even a mostly analog commuting quantum device can solve some classically intractable problems \cite{bermejo2018architectures}. We focus on an even more restricted model that should be easier to realize, in which the Hamiltonian $H$ is a specific commuting Hamiltonian containing only nearest-neighbor $Z$ operators, as discussed further below.

We also assume the ability to perform a globally controlled $CZ$ gate.
This arguably makes our model less analog, but it plays a key role in developing a sample-efficient protocol to verify the solutions given by the device, and it can potentially be realized using
experimental capabilities that have already been demonstrated \cite{VuleticAllOpticalSwitch, WaksPhotonControlledSwitch}, as we discuss in \cref{sec:experiment}.

\begin{definition}[Mostly analog + $\mathrm{GCZ}$ commuting quantum computation]
\label{defn:mostlyG}
	A mostly analog commuting model of computation is called \emph{mostly analog + $\mathrm{GCZ}$} if the system also contains a quantum degree of freedom (e.g., a qubit) that can serve as a global control for all of the qubits, such that one can apply---only $O(1)$ times---a global $CZ$ gate that is controlled by the degree of freedom and acts on all of the qubits. Here $\mathrm{GCZ}$ stands for global $CZ$.
\end{definition}

\textbf{The Classically Hard Problem.} 
In the protocol, the verifier asks the prover to solve a classically hard problem based on Hamiltonian evolution. The prover generates a quantum state but is not trusted to do so correctly. However, the prover is trusted to honestly measure this state to generate a classical witness. 
The verifier checks this witness to determine if the problem has been successfully solved. If so, then quantum advantage has been demonstrated.

Instead of considering a general quantum circuit, we aim to demonstrate quantum advantage by verifying a specific analog quantum simulation performed on a mostly analog + $\mathrm{GCZ}$ commuting machine. The simulation is motivated by the class of $\mathsf{IQP}$ (instantaneous quantum polynomial-time) circuits \cite{shepherd2009temporally,bremner2011classical}, in which all quantum gates are commuting (and thus interchangeable in time). Despite this strong restriction, $\mathsf{IQP}$ circuits are believed to be hard to simulate classically \cite{bremner2011classical,bremner2016average}. Furthermore, \textcite{bermejo2018architectures} presented a concrete scheme to show quantum speedup on an analog simulator by running a specific unit-time Hamiltonian evolution. The Hamiltonian includes only nearest-neighbor $ZZ$ interactions and local $Z$ terms (a form that we call a \emph{($ZZ+Z$)-type Hamiltonian}) on a 2-dimensional square lattice: 
\begin{equation}
\label{eq:bermejo-vega-Hamiltonian}
    \sum_{\{i,j\} \in \mathrm{NN} } \frac{\pi}{4} Z_i Z_j - \sum_{i=1}^n \frac{\pi}{4} Z_i,
\end{equation}
where $\mathrm{NN}$ denotes the set of edges connecting nearest-neighbor qubits.
The qubits are randomly initialized in either $\smash{\frac{1}{\sqrt2}}(\ket{0}+\ket{1})$ or $\smash{\frac{1}{\sqrt2}}(\ket{0}+e^{-\imag \pi/4}\ket{1})$.
\textcite{bermejo2018architectures,ringbauer2023verifiable} prove that a classical computer cannot efficiently sample from the output distribution of $X$-basis measurements on the above system within total variation distance (TVD) $0.292$, under plausible computational assumptions that we review in \cref{supp:conjectures}.
Moreover, since single-qubit $Z_i$ operators commute with all $Z_i Z_j$ operators, one can absorb the single-qubit evolution $\exp(\imag \frac{\pi}{4} \sum_i Z_i)$ into the initial state of each qubit, so that the qubits are initialized in either $\frac{1}{2}\left[ (1+\imag)\ket{0} + (1-\imag)\ket{1} \right]$ or $\frac{1}{2}\left[ (1+\imag)\ket{0} + e^{-\imag \pi/4} (1-\imag)\ket{1} \right]$, which can be prepared by single-qubit operations. Then the Hamiltonian $H$ to be simulated contains only $ZZ$ interaction terms:
\begin{equation}
\label{eq:dominik_hamiltonian}
	H = \sum_{\{i,j\} \in \mathrm{NN} } \frac{\pi}{4} Z_i Z_j. 
\end{equation}

\begin{table*}[t]
\centering
\begingroup
\setlength{\tabcolsep}{3pt}
\renewcommand{\arraystretch}{1.75}
\begin{tabular}{ |c||c|c|c|c|  }
 \hline
Scheme & \makecell[c]{\# of\\Measurements} & \makecell[c]{Classical\\Verification}  & \makecell[c]{Requirements for\\ Honest Provers}&\makecell[c]{Requirements for\\ Verifiers} \\
 \hline\hline
Cryptographic PoQs \cite{brakerski2021cryptographic,kahanamoku2022classically}  &$\poly(\lambda)$  &$\poly(\lambda)$ & Digital & Purely Classical \\ \hline
Random Circuit Sampling \cite{arute2019quantum, zhu2022quantum} & $O(\lambda)$ & $\exp(\lambda)$ & Digital & Purely Classical \\ \hline
Parent Hamiltonians \cite{bermejo2018architectures} & $O(\lambda^6)$ & $O(\lambda^6)$ &Analog & Single-Qubit Measurements\\ \hline
This Work (State Transmission) & $O(\lambda^2)$ & $O(\lambda^2)$ &Mostly Analog + Global $CZ$ & Single-Qubit Measurements\\ \hline
This Work (Trusted Measurements) & $O(\lambda^2)$ & $O(\lambda^2)$ & \makecell[c]{Mostly Analog + Global $CZ$ \\ + Trusted Measurements} & Purely Classical\\
\hline
\end{tabular}
\endgroup
\caption{Comparison of demonstrations of quantum advantage. As discussed in the main text, $\lambda$ denotes the security parameter.
}
\label{tabl1}
\end{table*}

\textbf{Main Result.}
We now have all the building blocks to formalize the main result. In the state-transmission scenario, we have the following theorem.
\begin{theorem}[Main result---state-transmission version]
    There exists a classically intractable sampling problem that can be verified by a single-round $\mathsf{QPIP}_1$ protocol where the prover runs a specific mostly analog + $\mathrm{GCZ}$ commuting quantum task $O(1)$ times.
\end{theorem}

In the trusted-measurement scenario, our result is as follows.

\begin{theorem}[Main result---trusted-measurement version]
    There exists a classically intractable sampling problem that can be verified by a single-round protocol where the classical verifier trusts the prover to perform single-qubit measurements, and the prover runs a specific mostly analog + $\mathrm{GCZ}$ commuting quantum task $O(1)$ times.
\end{theorem}

Our protocol has constant \emph{sample complexity}, i.e., it only requires the prover to generate $O(1)$ samples of an $n$-qubit state.
This is significantly less expensive than \textcite{bermejo2018architectures}, which uses $O(n^2)$ samples.

In both this work and Ref.~\cite{bermejo2018architectures}, the prover is expected to perform trusted measurements (or the prover sends qubits to the verifier for her to measure), unlike proofs of quantumness (PoQs) based on trapdoor claw-free functions (TCFs) \cite{brakerski2021cryptographic,kahanamoku2022classically} and quantum supremacy experiments \cite{arute2019quantum, preskill2012quantum} based on sampling problems, which makes it difficult to compare the resource requirements. 
However, in all of these schemes, single-qubit measurements must be performed many times, either by the prover or the verifier. Hence the number of qubits measured is a comparable quantity.

Equivalently, without transforming the protocols, we can still compare the number of measurements in terms of the \emph{security parameter}, whether the measurements are trusted or not.
The security parameter $\lambda$ is defined such that a dishonest prover without quantum computational power needs time $2^{\Omega(\lambda)}$ in order to make the verifier accept.
For our protocol, the number of qubits $n$ is quadratic in $\lambda$, just as in \textcite{bermejo2018architectures}. 
Under optimistic assumptions, cryptographic PoQs can probably have $n=O(\lambda)$ \cite{kahanamoku2022classically}, but for 
most common TCFs, $n$ scales at least quadratically with $\lambda$ \cite{brakerski2021cryptographic}. Since it has constant sample complexity, our protocol uses $O(\lambda^2)$ single-qubit measurements. This is
better than \textcite{bermejo2018architectures}, which uses $O(\lambda^3)$ measurements.
Furthermore, our protocol can be verified by $O(\lambda^2)$-time classical computation, significantly below the verification cost of $O(\lambda^6)$ for \textcite{bermejo2018architectures} and presumably $\exp(\lambda)$ for quantum supremacy experiments based on sampling problems \cite{arute2019quantum, boixo2018characterizing, dalzell2020many}.

We summarize the comparison between our work and other quantum advantage protocols in \cref{tabl1}.

On the prover side, TCF-based PoQs generally require $\mathrm{poly}(\lambda)$-depth low-noise digital quantum computation, while our honest strategy is designed for analog quantum simulators with only limited digital capabilities. This may be harder than fully analog simulation \cite{bermejo2018architectures, arute2019quantum, zhong2020quantum, wu2021strong}, but should still be feasible in the relatively near term. Moreover, our protocol can detect---and is robust against---a specific type of phase error that happens frequently in practice. Thus we believe our work achieves a significant improvement in terms of verification efficiency for verified quantum advantage protocols, and is an easier-to-implement scheme. We provide exact threshold fidelities (independent of the system size) for the device to demonstrate quantum advantage using our scheme. We also show that when the noise is incoherent, the fidelity requirements can be further relaxed.

The remainder of this paper is organized as follows.
In \cref{sec:qaprotocol}, we describe the sample-efficient quantum advantage protocol and analyze its resource requirements. In \cref{sec:honest}, we give the near-term strategy for honest provers and discuss potential experimental realizations. Finally, we summarize the results and discuss their implications and potential future extensions in \cref{sec:summary}.

\section{The Quantum Advantage Protocol}
\label{sec:qaprotocol}
\subsection{The Single-Step Feynman-Kitaev Construction}

Our protocol is inspired by the Feynman-Kitaev mapping \cite{kitaev2002classical}, which converts the task of executing a quantum circuit to that of finding the ground state of an associated Hamiltonian. The Feynman-Kitaev Hamiltonian is the foundation of several verification schemes in the circuit model: if a quantum server can provide the ground state (the witness) to the client, then the client can verify the quantum computation by measuring its energy. Examples include the \textcite{fitzsimons2018post} protocol where the prover needs to perform single-qubit trusted measurements, and the \textcite{mahadev2018classical} protocol that works even for untrusted measurements.

Inspired by the above protocols for circuit-model computations, we consider using a simplified Feynman-Kitaev mapping to verify analog quantum simulation of the system in \cite{bermejo2018architectures}, i.e., the Hamiltonian $H$ in \cref{eq:dominik_hamiltonian}.

We define the (single-step) \emph{history state}
\begin{equation}
\label{eq:single_original}
    \ket{\psi_\mathrm{hist}} = \frac{1}{\sqrt{2}} \left( \ket{0} \ket{\phi_\mathrm{in}} + \ket{1} U \ket{\phi_\mathrm{in}}\right),
\end{equation}
where $\ket{\phi_\mathrm{in}}$ is the input state and $U$ is the \emph{propagation unitary}. 
The state $\ket{\psi_\mathrm{hist}}$ is the ground state of the single-step Feynman-Kitaev Hamiltonian. Since we are considering quantum simulation of the $ZZ$-type Hamiltonian $H$ defined in \cref{eq:dominik_hamiltonian}, we have $U = \exp(-\imag H T)$ with $T=1$, and $\ket{\phi_\mathrm{in}}$ is the same random input state defined in the system of \textcite{bermejo2018architectures} with single-qubit $Z$ evolution absorbed. The computationally hard sampling problem can be solved by measuring $U\ket{\phi_\mathrm{in}}$ in the $X$ basis. We use $P_\mathrm{ideal}$ to denote the ideal distribution of measurement outcomes.

The Feynman-Kitaev Hamiltonian includes a term
\begin{equation}
    H^\mathrm{prop} = \frac{1}{2} \left(I \otimes I - \ket{1}\bra{0} \otimes U- \ket{0}\bra{1} \otimes U^\dagger \right),
\end{equation}
which ensures that the ground state encodes the correct propagation unitary $U$.
One can easily check that $H^\mathrm{prop}$ is positive semidefinite and $H^\mathrm{prop}\ket{\psi_\mathrm{hist}}=0$, so $\ket{\psi_\mathrm{hist}}$ is a ground state of $H^\mathrm{prop}$.

The other term of the Feynman-Kitaev Hamiltonian is
\begin{equation}
    H^\mathrm{in}=\ket{0}\bra{0} \otimes  \left(I- \sum_i\ket{\phi_{\mathrm{in},i}}\bra{\phi_{\mathrm{in},i}} \right),
\end{equation}
where $\ket{\phi_{\mathrm{in},i}}$ is the state of the $i$th qubit of $\ket{\phi_\mathrm{in}}$. $H^\mathrm{in}$ ensures that the input state is $\ket{\phi_\mathrm{in}}$. It is also positive semidefinite and satisfies $H^\mathrm{in}\ket{\psi_\mathrm{hist}}=0$.

A toy version of our protocol for demonstrating quantum advantage, without any technical detail, is as follows. The verifier sends classical descriptions of $H$ and $\ket{\phi_\mathrm{in}}$ to the prover, and asks the prover to prepare $N_\mathrm{M}$ copies of the history state $\smash{\frac{1}{\sqrt2}}( \ket{0}\ket{\phi_\mathrm{in}} + \ket{1}U\ket{\phi_\mathrm{in}})$.
For each copy, the verifier chooses whether to generate a sample or to verify the state, with equal probability. If she chooses to sample, then she asks the prover to measure the first qubit (i.e., the clock qubit) in the $Z$ basis and all other qubits in the $X$ basis, and a sample is generated if the first measurement outcome is $-1$ (i.e., the clock qubit is in $\ket{1}$). If the verifier chooses to verify, then she measures the energy of $H^\mathrm{prop}+H^\mathrm{in}$ by quantum phase estimation. Finally, if every run of quantum phase estimation returns $0$, which means that the fidelity between the measured state and the perfect history state is very high (the infidelity is inverse exponential in $N_\mathrm{M}$ if $N_\mathrm{M}/2$ copies are chosen for verification) and therefore the measurement outcomes are close to the desired distribution $P_\mathrm{ideal}$, she accepts and announces all of the samples obtained. Otherwise, she rejects.

One disadvantage of the verification part of this scheme is that it can only accept devices that provide history states with exponentially small infidelity. 
While near-term devices will be imperfect, they might still be able to sample from classically intractable distributions. Also, experimentalists may prefer to know how well their devices are performing and whether they are making progress, but a ``yes or no" result cannot provide this kind of information. Finally, measurements of $H^\mathrm{prop}+H^\mathrm{in}$ might be difficult, potentially requiring many measurements to determine the energy with sufficiently high precision, and quantum phase estimation is not feasible in the near term.

Therefore, inspired by the original single-step Feynman-Kitaev Hamiltonian, we propose a new verification scheme to replace the toy protocol. In the new scheme, different \emph{parameters} are measured to lower bound the total variation distance between the sampled distribution $P_\mathrm{exp}$ and the desired distribution $P_\mathrm{ideal}$, demonstrating quantum advantage according to \cite{bermejo2018architectures,ringbauer2023verifiable}. We also give an efficient near-term strategy for estimating those parameters.

\subsection{Our Measurement Scheme} 

To begin, consider an arbitrary $(n+1)$-qubit state
\begin{equation}\label{eqn-rho-diag}
    \rho = \sum_{i} p_i \ket{\psi_i}\bra{\psi_i},
\end{equation}
where $\{\ket{\psi_i}\}$ is the (unknown) eigenbasis of $\rho$, and $p_i$ is the probability corresponding to $\ket{\psi_i}$. We can write
\begin{equation}\label{eqn-psi-diag}
    \ket{\psi_i} = \alpha_i \ket{0} \ket{\phi_i} + \beta_i \ket{1} \ket{\phi'_i}
\end{equation}
where $\ket{\phi_i}$ and $\ket{\phi'_i}$ are $n$-qubit states and $\alpha_i,\beta_i \in \mathbb{C}$ satisfy $|\alpha_i|^2 + |\beta_i|^2 = 1$. Thus we can interpret $\rho$ as a classical mixture of states $\ket{\psi_i}$ as above with input states $\ket{\phi_i}$ and output states $\ket{\phi'_i}$.

The first parameter to be estimated in our scheme, the \emph{input fidelity}, is defined as
\begin{equation}
    F_\mathrm{in}(\rho) := \frac{ \sum_i p_i |\alpha_i|^2 \left| \langle \phi_i | \phi_\mathrm{in}  \rangle \right|^2} {\sum_i p_i |\alpha_i|^2}.
\end{equation}
This quantifies the quality of initial state preparation. It plays a similar role to $\langle H^\mathrm{in}\rangle$ in the single-step Feynman-Kitaev Hamiltonian.

Another parameter is the probability of obtaining a $-1$ outcome when measuring the clock qubit. We call this the \emph{probability of sampling}:
\begin{equation}
    p_\mathrm{samp} := \sum_i p_i |\beta_i|^2.
\end{equation}

The last parameter is the \emph{de facto} mean value of the non-Hermitian operator
\begin{equation}
O_{10}:=\ket{1}\bra{0}\otimes U,
\end{equation}
whose expectation value in the state $\rho$ is
\begin{equation}
        \Tr[\rho O_{10}] = \sum_i p_i \alpha_i \beta^*_i \langle \phi'_i|U|\phi_i\rangle.
\end{equation}
We mainly consider its squared norm, $|{\Tr[\rho O_{10}]}|^2$. This quantity is related to the quality of propagation from $\ket{\phi}$ to $U\ket{\phi}$, so it plays a similar role to $\langle H^\mathrm{prop} \rangle$ in the single-step Feynman-Kitaev Hamiltonian.

As we show in \cref{lemma:mmtsufficient_fin,lemma:mmtsufficient}, $F_\mathrm{in}(\rho)$, $p_\mathrm{samp}$, and $\Tr[\rho O_{10}]$ can all be estimated by single-qubit measurements, and the precision depends only on the number of samples measured, independent of the system size. Note that $O_{10}$ is not Hermitian, so it is in general not an observable, but its \emph{de facto} mean value (which is a complex number) can still be estimated. We discuss this in detail in the proof of \cref{lemma:mmtsufficient}.

We are interested in the \emph{output fidelity}, defined as
\begin{equation}
    F_\mathrm{output}:= \frac{\sum_i p_i |\beta_i|^2 |\langle \phi'_i|U|\phi_\mathrm{in} \rangle|^2}{\sum_i p_i |\beta_i|^2}.
\end{equation}
This quantifies the fidelity between the state being measured to generate samples from $P_\mathrm{exp}$ and the ideal state that can be measured to generate samples from $P_\mathrm{ideal}$, and thus can be directly related to the TVD between distributions, $\TVD(P_\mathrm{exp},P_\mathrm{ideal})$. In \cref{supp:relate_yikai}, we explicitly relate $F_\mathrm{output}$ and $\TVD(P_\mathrm{exp}, P_\mathrm{ideal})$, and find the threshold fidelity $0.915$ using the hardness result proved in \cite{ringbauer2023verifiable}, which gives a criterion for verified quantum advantage.

In \cref{supp:relate_yikai}, we also derive a lower bound for $F_\mathrm{output}$ in terms of $\epsilon := 1/4 - \Tr[\rho O_{10}]$, $\epsilon' := 1/2- p_\mathrm{samp}$, and $\epsilon'' := 1- F_\mathrm{in}(\rho)$, as follows.
\begin{theorem}[Lower bound on the output fidelity]
\label{thm:lowerbounddistance}
\begin{equation}
    F_\mathrm{output} \geq 1-16\epsilon - 3\epsilon''+\mathrm{h.o.}
\end{equation}
where $\mathrm{h.o.}$ indicates higher-order terms in $\epsilon,\epsilon',\epsilon''$.
\end{theorem}
If the device is close to perfect (which is the scenario we consider here), then $\epsilon,\epsilon'' \ll 1$ and $|\epsilon'|\ll 1$. Hence, the higher-order terms can be safely dropped, as is shown in detail in \cref{supp:relate_yikai}, and the above bound can be written as
\begin{equation}
    F_\mathrm{out}(\rho) \geq 16 |{\Tr[\rho O_{10}]}|^2 + 3F_\mathrm{in}(\rho) - 6.
\end{equation}

Using \cref{thm:lowerbounddistance} with threshold fidelity $0.915$, we conclude that the measurement outcomes sample from a classically intractable distribution provided $4|{\Tr[\rho O_{10}]}|^2 \geq 0.988$ and $F_\mathrm{in}(\rho) \geq 0.988$.

Observe that the final lower bound does not contain first-order terms in $\epsilon' = 1/2 - p_\mathrm{samp}$. 
However, we still need to estimate $p_\mathrm{samp}$ to ensure that its value is sufficiently close to $1/2$ that our first-order approximation still holds. Hence, we also require $|1/2-p_\mathrm{samp}|\leq 0.012$.

It is clear from the above theorem that our protocol can also tolerate a small amount of noise in the measurements of the quantum state. To simplify the analysis, in the rest of this section, we make the perfect-measurement assumption: all measurements, whether performed by the prover in the trusted-measurement scheme or by the verifier in the state-transmission scheme, are noiseless. We postpone the discussion of noisy measurements to \cref{supp:noisy}.

We claim that the number of copies of the history state needed to verify quantumness (i.e., the sample complexity) depends only on the precision and is not related to the system size $n$. As a result, the prover only needs to perform $O(n)$ trusted single-qubit measurements. These properties are formalized and proven in \cref{lemma:mmtsufficient_fin,lemma:mmtsufficient}.

Since the TVD between ideal and real output distributions is lower bounded by estimating $F_\mathrm{in}$ and $\Tr[\rho O_{10}]$, the sample complexity of the protocol is determined by how many copies of the state are required to estimate both quantities to a specific precision.

\begin{lemma}[Sufficiency of single-qubit measurements for $F_\mathrm{in}$ and $p_\mathrm{samp}$]
\label{lemma:mmtsufficient_fin}
    A verifier capable of single-qubit measurements and polynomial-time classical computation can estimate $F_\mathrm{in}$ and $p_\mathrm{samp}$ in a mixed state $\rho$ with error at most $\delta_o$ using $O(1/\delta_o^2)$ samples of $\rho$.
\end{lemma}
\begin{proof}
First recall that the ideal input state $\ket{\phi_\mathrm{in}}$ is a product state of either $\ket{x}:=\frac{1}{2}\left[ (1+\imag)\ket{0} + (1-\imag)\ket{1} \right]$ or $\ket{y}:=\frac{1}{2}\left[ (1+\imag)\ket{0} + e^{-\imag \pi/4} (1-\imag)\ket{1} \right]$. Their corresponding orthogonal states are $\ket{x^\perp}:=\frac{1}{2}[(1-\imag)\ket{0}-(1+\imag)\ket{1}]$ and $\ket{y^\perp}:=\frac{1}{2}[(1-\imag)\ket{0}-e^{-\imag \pi/4}(1+\imag)\ket{1}]$, respectively.

If a pure state $\ket{\psi_i}=\alpha_i \ket{0} \ket{\phi_i} + \beta_i \ket{1} \ket{\phi'_i}$ is given, the fidelity of the input state, $|\langle \phi_i|\phi_\mathrm{in}\rangle|^2$, can be estimated as follows. We first measure the clock qubit in the $Z$ basis, and if the outcome is $+1$ (so the state collapses to $\ket{0}\ket{\phi_i}$), we measure every other qubit in its corresponding rotated basis, which is either $\{\ket{x},\ket{x^\perp}\}$ or $\{\ket{y},\ket{y^\perp}\}$. If all measurement outcomes are $+1$, then $\ket{\phi'_i}$ collapses to $\ket{\phi_\mathrm{in}}$. Therefore, if the number of copies for which the clock qubit measurement gives $+1$ is $N_\mathrm{in+}$, and among them the number of copies where all other measurements give $+1$ is $N_\mathrm{in+0}$, then $\frac{N_\mathrm{in+0}}{N_\mathrm{in+}}$ is an unbiased estimator of $|\langle \phi_i | \phi_\mathrm{in} \rangle|^2$. Furthermore, for a mixed state $\rho$, the same strategy gives an estimate of $F_\mathrm{in}(\rho)$:
\begin{equation}
     F_\mathrm{in}(\rho) = \lim _{N_{\mathrm{in+}} \rightarrow \infty} \frac{N_\mathrm{in+0}}{N_\mathrm{in+}}.
\end{equation}
The precision of estimating $F_\mathrm{in}$ increases with $N_\mathrm{in+}$.
More precisely, we can use Hoeffding's inequality to quantify their relationship:
\begin{equation}
    \Pr[|F_\mathrm{in,M}-F_\mathrm{in}|\geq \delta_o] \leq 2\exp(-2\delta_o^2 N_\mathrm{in+}),
\end{equation}
where $F_\mathrm{in,M}$ represents the estimate from measurements. For the estimate of $F_\mathrm{in}$ to have error at most $\delta_o$ with probability at least $1-p_e$, it suffices to use $N_\mathrm{in+} = O(|{\ln p_e}|/\delta_o^2)$ valid measurements, independent of the system size. Moreover, since the single-step history state has equal weight between the $\ket{0}$ and $\ket{1}$ states of the clock qubit, $N_\mathrm{in+}$ should be close to $N_\mathrm{M}/2$, where $N_\mathrm{M}$ is the total number of states measured.

We also describe how to estimate $p_\mathrm{samp}$. Fortunately, this can already be obtained from $N_\mathrm{in+}$. Since $p_\mathrm{samp}$ is just the probability of a $Z$-basis measurement of the first qubit returning $-1$, 
$\frac{N_\mathrm{in+}}{N_M}$ is an unbiased estimator of $p_\mathrm{samp}$. Similarly, the probability for the estimate of $p_\mathrm{samp}$ to have error more than $\delta_o$ is upper bounded as
\begin{equation}
    \Pr[|p_\mathrm{samp,M}-p_\mathrm{samp}|\geq \delta_o] \leq 2\exp(-2\delta_o^2 N_\mathrm{M}),
\end{equation}
where $p_\mathrm{samp,M}$ denotes the estimated value of $p_\mathrm{samp}$.
Since $N_\mathrm{M}>N_\mathrm{in+}$, we can always estimate $p_\mathrm{samp}$ to a higher precision than $F_\mathrm{in}$ when they are estimated together.

\end{proof}

\begin{lemma}[Sufficiency of single-qubit Pauli measurements for $|\langle O_{10}\rangle|^2$]
\label{lemma:mmtsufficient}
A verifier capable of single-qubit measurements and polynomial-time classical computation can estimate $|\langle O_{10}\rangle|^2$ in a mixed state $\rho$ with error at most $\delta_o$ using $O(1/\delta_o^2)$ samples of $\rho$.
\end{lemma}
\begin{proof}
We can write
\begin{equation}
    O_{10}=\ket{1}\bra{0}\otimes U = \frac{1}{2}(X-\imag Y)\otimes U.
\end{equation}
It can be difficult to measure $O_{10}$ in general, because $U$ typically decomposes into exponentially many Pauli terms. Fortunately, in our protocol, we have $U=\exp(-\imag H T)$ for the $ZZ$-type Hamiltonian
\begin{equation}
    H=\frac{\pi}{4} \sum_{k=1}^m  H_k = \frac{\pi}{4} \sum_{\{i,j\} \in \mathrm{NN} }  Z_i Z_j ,
\end{equation}
where each $H_k$ is one of the $Z_i Z_j$s. As all $H_k$ terms commute, we can decompose $U$ into a product of evolutions for each term, and further express these evolutions in terms of trigonometric functions as every $H_k$ is a Pauli string:
\begin{equation}
\label{eq:Us_decomposed}
\begin{aligned}
    U & = \exp\left(-\imag\frac{\pi}{4} \sum_{k=1}^m  H_k \right) \\
    & = \prod_{k=1}^m \exp\left(-\imag \frac{\pi}{4} H_k   \right)\\
    & = \prod_{k=1}^m \left( \cos(\frac{\pi}{4})I -\imag \sin(\frac{\pi}{4}) H_k \right).
\end{aligned}
\end{equation}
$U$ is not a well-defined quantum observable since it is not Hermitian, but we can still define its \emph{de facto} single-measurement outcome as a complex number $u$. Since all $H_k$s can be simultaneously measured, $u$ can be inferred by evaluating the right-hand side of \cref{eq:Us_decomposed}. More specifically, letting $h_k$ denote the outcome of a single measurement of $H_k$, we have
\begin{equation}
\label{eq:evaluation_equation}
    u = \prod_{k=1}^m \left( \cos(\frac{\pi}{4}) -\imag \sin(\frac{\pi}{4}) h_k \right).
\end{equation}
Since each $H_k$ is $Z_i Z_j$, the verifier need only perform single-qubit $Z$ measurements to obtain the $h_k$s.

In summary, to estimate the expected value of $O_{10}$, it suffices to measure the clock qubit in either the $X$ or the $Y$ basis, measure all other qubits in the $Z$ basis to get the values of $u$, and repeat this process enough times to obtain the mean values of $X\otimes U$ and $Y\otimes U$ with sufficiently high precision.

To determine the number of samples required, we evaluate the probability that the measured value deviates from the expected value using concentration bounds. Note that $O_{10}$ is not Hermitian, so its \emph{de facto} measurement outcomes are complex numbers. Recall that $O_{10} = \frac{1}{2} \left( X\otimes U-\imag Y\otimes U\right)$, so one sample of the value of $O_{10}$ can be obtained by measuring two copies of the state of interest, and both the real and imaginary parts of the measurement outcome of $O_{10}$ are at most $1/2$. Therefore, for any $0<\delta_o <1/2$, letting $\langle \cdot \rangle_\mathrm{M}$ be the average of the measurement outcomes after running the experiment $N_\mathrm{M}$ times, and using Hoeffding's inequality,
\begin{equation}
\label{eq:hoeff_main}
\begin{aligned}
    &\Pr\left[ \left| |\langle O_{10}\rangle_\mathrm{M}|^2 - |\langle O_{10} \rangle|^2 \right| \geq \delta_o^2 \right]\\
    &\quad \leq \Pr[ |\Re[\langle O_{10}\rangle_\mathrm{M}] - \Re[\langle O_{10} \rangle]| \geq \delta_o ] \\
    &\qquad + \Pr[ |\Im[\langle O_{10}\rangle_\mathrm{M}] - \Im[\langle O_{10} \rangle]| \geq \delta_o ]\\
    &\quad \leq 4\exp(- 2 \delta_o^2 N_\mathrm{M}).
\end{aligned}
\end{equation}

In conclusion, to ensure that the error in the estimation of $|\langle O_{10} \rangle|^2$ is less than $\delta_o$ with probability at least $1-p_e$,
it suffices to measure $O_{10}$ on $2N_\mathrm{M} = O\left(|\ln p_e| / \delta_o^2 \right)$ copies of the state, irrespective of the size of the system. Moreover, if $p_e$ is a negligible function of
the security parameter $\lambda$, then $N_\mathrm{M}$ only needs to scale linearly with $\lambda$. In other words, the probability of obtaining a wrong estimate of $|\langle \Tr[\rho O_{10}] \rangle|^2$ converges to 0 exponentially with respect to the number of copies, $N_\mathrm{M}$.
\end{proof}

\subsection{Our Protocol}
In this subsection, we outline the behavior of the verifier and the prover in our protocol, and present the soundness and completeness conditions.

The verifier first provides the prover with descriptions of $H$ and $\ket{\phi_\mathrm{in}}$, and the desired number of copies of the history state $N_\mathrm{M}$ (whose value is determined in \cref{thm:complete,thm:soundness}).

The verifier asks the prover to perform measurements to estimate (or measures by herself if state transmission is allowed) $\left|\langle O_{10}\rangle \right|^2$, $N_\mathrm{in+0}$, and $N_\mathrm{in+}$ from the $N_\mathrm{M}$ samples to verify the correctness of the output state. She also asks the prover to generate samples by measuring the $\ket{\phi'}$ state conditioned on obtaining $-1$ from measuring the clock state.
Therefore, the verifier should generate two random bits for every state before measuring it.

The first bit, $b_\mathrm{sampling}$, determines whether the verifier should ask the prover to generate samples ($b_\mathrm{sampling}=1$) or verify the output state ($b_\mathrm{sampling}=0$). If $b_\mathrm{sampling}=1$, the prover should measure the clock qubit in the standard basis and all system qubits in the Hadamard basis. If the clock is measured to be $-1$, and if the prover passes the verification protocol, then the outcomes of Hadamard measurements on system qubits are samples from the desired distribution.

When $b_\mathrm{sampling}=0$, the verifier must decide whether to use this copy to estimate $|\langle O_{10} \rangle|^2$ or $F_\mathrm{in}(\rho)$ and $p_\mathrm{samp}$ by generating the other random bit $b_\mathrm{testtype}$. If the second random bit, $b_\mathrm{testtype}$, is $0$, then she estimates $F_\mathrm{in}(\rho)$ and $p_\mathrm{samp}$ by asking the prover to measure the clock qubit in the computational basis and all system qubits in their corresponding basis, updating the values of $N_\mathrm{in+0}$ and $N_\mathrm{in+}$, as in the proof of \cref{lemma:mmtsufficient_fin}. For $b_\mathrm{testtype}=1$, she estimates $|\langle O_{10} \rangle|^2$, so the prover should use the same strategy as in the proof of \cref{lemma:mmtsufficient} to measure the value of $U$ and, subsequently, the values of $X\otimes U$ or $Y\otimes U$.

In the end, the verifier estimates the parameters of interest. As in the proofs of \cref{lemma:mmtsufficient_fin,lemma:mmtsufficient}, we denote the estimated values of $|\langle O_{10}\rangle|^2$, $p_\mathrm{samp}$, and $F_\mathrm{in}$ by $|\langle O_{10}\rangle_\mathrm{M}|^2$, $p_\mathrm{samp,M}$, and $F_\mathrm{in,M}$, respectively.
The verifier then decides to accept or not by checking whether the estimated values are within the acceptance ranges, which are $0.994\leq 4|\langle O_{10} \rangle_\mathrm{M}|^2 \leq 1$, $0.994\leq F_\mathrm{in,M}(\rho)\leq 1$, and $0.494\leq p_\mathrm{samp,M} \leq 0.506$.
Note that here we choose more stringent values than the quantum advantage criterion in \cref{thm:lowerbounddistance} such that if the fidelity of the output state is slightly below the quantum advantage criterion, the verifier will reject with high probability. This is related to the \emph{soundness} of the protocol, which is discussed in detail in Theorem~\ref{thm:soundness}.

\protocol{Protocol to demonstrate quantum advantage by analog quantum simulation}{Our protocol for demonstrating quantum advantage.}{prot:protocol_draft_1}{
Let $H$ be a Hamiltonian to be simulated,
and let $\ket{\phi_\mathrm{in}}$ be the initial state of Hamiltonian evolution. 

\begin{enumerate}
\item The verifier initializes counters $s_{\{X,U\}}, s_{\{Y,U\}}, N_X, N_Y, N_\mathrm{in+}, N_\mathrm{in+0}$ to 0. She sends $N_M$ and classical descriptions of $H$ and $\ket{\phi_\mathrm{in}}$ to the prover.
\item The prover creates $N_\mathrm{M}$ copies of the correct history state $\ket{\psi} = \frac{1}{\sqrt{2}}(\ket{0}\ket{\phi_\mathrm{in}} + e^{i\gamma}\ket{1}U\ket{\phi_\mathrm{in}})$, where $\gamma$ is a fixed arbitrary phase, and (only in the state-transmission scenario) sends them to the verifier.
\item For each state (in the trusted-measurement scenario) to be measured by the prover or (in the state-transmission scenario) to be received by the verifier:
\begin{enumerate}    
    \item The verifier generates 2 random bits $b_\mathrm{sampling}$ and $b_\mathrm{testtype}$. If $b_\mathrm{sampling}=1$, the verifier will obtain a sample from the distribution. If $b_\mathrm{sampling}=0$ and $b_\mathrm{testtype}=0$, the verifier will check if the input state is correct. If $b_\mathrm{sampling}=0$ and $b_\mathrm{testtype}=1$, the verifier will check if the Hamiltonian evolution is correct.
    
    \item If $b_\mathrm{sampling}=1$, the verifier measures (or asks the prover to measure) the first qubit in the standard basis. If the measurement outcome is $-1$, then she measures all other qubits in the Hadamard basis and stores the measured bit string.

    \item If $b_\mathrm{sampling}=0$ and $b_\mathrm{testtype}=0$, the verifier measures (or asks the prover to measure) the first qubit in the standard basis. If the outcome is $+1$:
    \begin{enumerate}
        \item The verifier updates the counter by $N_\mathrm{in+}\leftarrow N_{\mathrm{in+}}+1$.
        \item The verifier measures (or asks the prover to measure) every other qubit in the following basis: if its initial state is supposed to be $\ket{x}$, then measure it in the $\{\ket{x},\ket{x^\perp}\}$ basis; otherwise, measure it in the $\{\ket{y} , \ket{y^\perp}\}$ basis.
        \item If all outcomes are $+1$, she updates the counter as $N_{\mathrm{in+}0}\leftarrow N_{\mathrm{in+}0}+1$.
    \end{enumerate}

    \item If $b_\mathrm{sampling}=0$ and $b_\mathrm{testtype}=1$:
    \begin{enumerate}
        \item The verifier selects the basis from $\{X,Y\}$ randomly, measures (or asks the prover to measure) the clock qubit in the chosen basis, and stores the measurement outcome in $b$.
        \item The verifier measures (or asks the prover to measure) all system qubits in the $Z$ basis. Then, she calculates the values of $U$ according to the proof of \cref{lemma:mmtsufficient}, denoted by $u$.
        \item If the basis chosen is $X$, the verifier updates the counters as $N_X\leftarrow N_X+1$, $s_{\{X,U\}}\leftarrow s_{\{X,U\}} + b u$.
        \item If the basis chosen is $Y$, the verifier updates the counters as $N_Y\leftarrow N_Y+1$, $s_{\{Y,U\}}\leftarrow s_{\{Y,U\}} + b u$.
    \end{enumerate}
\end{enumerate}
\item 
\begin{enumerate}
    \item The verifier calculates $h_{X,U}=s_{\{X,U\}} / N_X$ and $h_{Y,U}=s_{\{Y,U\}} / N_Y$. She also calculates $\langle O_{10}\rangle_\mathrm{M} = h_{X,U}-\imag h_{Y,U}$ and $4|\langle O_{10} \rangle_\mathrm{M}|^2$.
    \item The verifier calculates $F_\mathrm{in,M} = \frac{N_\mathrm{in+0}}{N_\mathrm{in+}}$.
\end{enumerate}
\item If $4|\langle O_{10} \rangle_\mathrm{M}|^2 > 0.988$ and $F_\mathrm{in,M} > 0.988$, the verifier accepts the interaction and publishes the stored bit strings as the samples from the distribution. Otherwise, she rejects.
\end{enumerate}
}

We now present the completeness and soundness properties of the protocol. A proof of quantumness is called \emph{complete} if any honest prover with quantum computational ability (which in our case means being able to prepare the required history state $\ket{\psi_\mathrm{hist}}$ with tolerable error, as explained in more detail below) is accepted with probability at least $2/3$. It is called \emph{sound} if no prover with only classical polynomial-time computational ability can be accepted with probability higher than $1/3$.

Before showing the completeness theorem, we observe that any phase error in the clock qubit does not affect the correctness of sampling, which means that a family of history states can be and should be accepted. In fact, one can easily check that $F_\mathrm{in}(\ket{\psi_\mathrm{hist}(\theta)}\bra{\psi_\mathrm{hist}(\theta)})=1$ and $4|{\Tr[\ket{\psi_\mathrm{hist}(\theta)}\bra{\psi_\mathrm{hist}(\theta)}O_{10}]}|^2=1$ for all $\ket{\psi_\mathrm{hist}(\theta)}:=\frac{1}{\sqrt{2}}(\ket{0}\ket{\phi_\mathrm{in}}+e^{\imag \theta}\ket{1}U\ket{\phi_\mathrm{in}})$, where $\theta$ can be any real number. This immediately leads to the following completeness result.

\begin{theorem}[Completeness]
\label{thm:complete}
    If the prover constructs $N_M=3.5\times 10^6$ copies of $\ket{\psi_\mathrm{hist}(\theta)}$ with a fixed value of $\theta$, then the verifier will accept with probability at least $2/3$.
\end{theorem}
\begin{proof}
    We can calculate that $F_\mathrm{in}(\ket{\psi_\mathrm{hist}(\theta)}\bra{\psi_\mathrm{hist}(\theta)})=1$, $4|{\Tr[\ket{\psi_\mathrm{hist}(\theta)}\bra{\psi_\mathrm{hist}(\theta)}O_{10}]}|^2=1$, and $p_\mathrm{samp}=1/2$. Therefore, it suffices to ensure the probabilities that the measurement errors exceed $0.0015$ for $|\langle O_{10} \rangle|^2$, and $0.006$ for $F_\mathrm{in,M}$ and $p_\mathrm{samp}$, are all less than $1/3$.
    
    Suppose that of $N_\mathrm{M}$ available samples, $N_\mathrm{M}/2$ are used to generate samples, $N_\mathrm{M}/4$ are used to estimate $|\langle O_{10} \rangle|^2$, and $N_\mathrm{M}/4$ are used to estimate $F_\mathrm{in}$. According to \cref{lemma:mmtsufficient_fin,lemma:mmtsufficient}, and letting $N_\mathrm{in+} = N_\mathrm{M}/8$, the probability of rejection is at most $\max\{2\exp(- 0.006^2 N_\mathrm{M}/4) , 4\exp(- 0.0015^2 N_\mathrm{M}/2)\} = 0.08<1/3$.
\end{proof}

However, in a real experiment, it is unlikely for a device to only make one specific error---a phase error on the clock qubit---and to otherwise produce $\ket{\psi_\mathrm{hist}(\theta)}$ perfectly. Instead, every experimental platform might have its own pattern of noise with multiple types of errors. Our verification scheme also has some robustness against these more general errors. Here we characterize the robustness for the case where the device can prepare a noiseless initial state but its Hamiltonian evolution has some error.

\begin{theorem}[Completeness + Robustness]
    If the prover constructs $N_M=3.5\times 10^{6}$ copies of the noisy history state $\ket{\psi_\mathrm{noisy}}:=\frac{1}{\sqrt{2}}(\ket{0}\ket{\phi_\mathrm{in}}+e^{\imag \theta}\ket{1}\ket{\phi'})$ where $|\bra{\phi'}U\ket{\phi_\mathrm{in}}|^2=0.999$, then the verifier will accept the interaction with probability at least $2/3$.
\end{theorem}

\begin{proof}
We can check that $F_\mathrm{in}(\ket{\psi_\mathrm{noisy}}\bra{\psi_\mathrm{noisy}})=1$, $p_\mathrm{samp}=1/2$, and
$4|\langle O_{10} \rangle|^2 = |\langle\phi'|U|\phi_\mathrm{in}\rangle|^2 = 0.999$. Therefore, it suffices to estimate $4|\langle O_{10} \rangle|^2$ within precision $0.005$ and $F_\mathrm{in}$ and $p_\mathrm{samp}$ within precision $0.006$. This precision can be achieved using $N_M$ copies of the prepared state, which gives success probability $0.73>2/3$.
\end{proof}

Next, we establish the soundness condition. Recall that, informally, a quantum advantage protocol is called \emph{sound} if all provers without quantum computational capability are rejected by the verifier with high probability.

\begin{theorem}[Soundness]
\label{thm:soundness}
    If the verifier accepts with probability at least $2/3$ with $N_\mathrm{M}=3.5\times 10^{6}$ copies of the state provided by the prover, then measurements of the state generate samples from a classically intractable distribution.
\end{theorem}

\begin{proof}
    This theorem has almost been proven in \cref{thm:lowerbounddistance}, in which $F_\mathrm{output} \geq 0.915$ is guaranteed if $F_\mathrm{in}\geq 0.988$, $|p_\mathrm{samp}-1/2|\leq 0.012$, and $4|\langle O_{10} \rangle|^2 \geq 0.988$. Also, according to the proof of \cref{thm:complete}, with $N_\mathrm{M}$ samples, the error in the estimation of all parameters is lower than $0.006$ with probability at least $2/3$.
    
    Therefore, if the verifier accepts with probability at least $2/3$, which means that $F_\mathrm{in,M}\geq 0.994$, $|p_\mathrm{samp}-1/2|\leq 0.006$, and $4|\langle O_{10} \rangle_\mathrm{M}|^2 \geq 0.994$ with probability at least $2/3$, then it is immediately clear that $F_\mathrm{in}\geq 0.988$, $|p_\mathrm{samp}-1/2|\leq 0.012$, and $4|\langle O_{10} \rangle|^2 \geq 0.988$, which implies that $F_\mathrm{output}\geq 0.915$.
\end{proof}

A detailed description of the protocol can be found in \cref{prot:protocol_draft_1}.

One hidden assumption in this section is that all copies of the history state provided by the prover are independent of each other. However, if the prover is an adversarial challenger, he can provide correlated states. In \cref{supp:martingale}, we outline how martingale inequalities can be used to show that our protocol is sound even if the states measured are correlated across multiple trials.

The analysis in this section assumes noiseless measurements, which are impractical in real devices. We discuss the protocol's tolerance of noisy measurements in \cref{supp:noisy}.

\section{The Honest-Prover Strategy}
\label{sec:honest}
\subsection{History State Preparation}
\label{sec:echo_trick}

\begin{figure}
    \centering
    \begin{tikzpicture}
    
    \filldraw (0,0) circle (3pt);
    \filldraw (0.5,0.5) circle (3pt);
    \filldraw (0.5,-0.5) circle (3pt);
    \filldraw (-0.5,0.5) circle (3pt);
    \filldraw (-0.5,-0.5) circle (3pt);
    \filldraw (1,1) circle (3pt);
    \filldraw (-1,1) circle (3pt);
    \filldraw (-1,-1) circle (3pt);
    \filldraw (1,-1) circle (3pt);
    \filldraw (0,-1) circle (3pt);
    \filldraw (0,1) circle (3pt);
    \filldraw (1,0) circle (3pt);
    \filldraw (-1,0) circle (3pt);
    
    \draw (0.5,0) circle (3pt);
    \draw (-0.5,0) circle (3pt);
    \draw (0,0.5) circle (3pt);
    \draw (0,-0.5) circle (3pt);
    \draw (1,0.5) circle (3pt);
    \draw (-1,-0.5) circle (3pt);
    \draw (0.5,1) circle (3pt);
    \draw (0.5,-1) circle (3pt);
    \draw (-0.5,1) circle (3pt);
    \draw (-0.5,-1) circle (3pt);
    \draw (-1,0.5) circle (3pt);
    \draw (1,-0.5) circle (3pt);
    \end{tikzpicture}
    \caption{The square lattice can be divided into two parts such that every $ZZ$ operator acts on qubits from both parts.}
    \label{fig:square}
\end{figure}
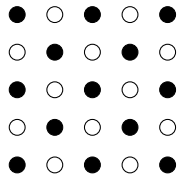

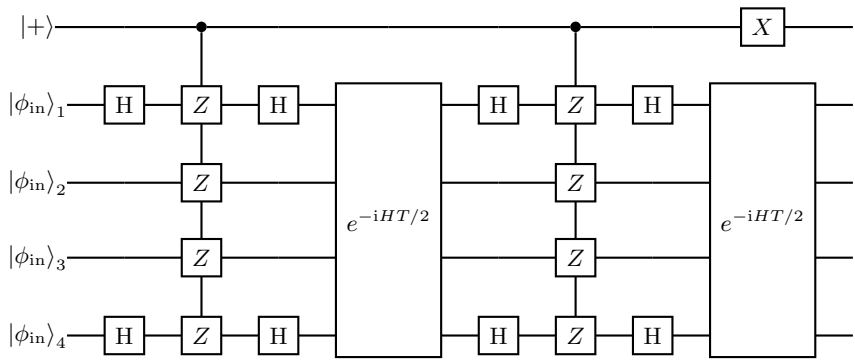
\begin{figure*}[t]
\centering
\begin{quantikz}
&\ket{+} &\qw &\ctrl{4}&\qw &\qw &\qw&\ctrl{4}&\qw &\gate{X} &\qw \\
&\ket{\phi_\mathrm{in}}_1 &\gate{\mathrm{H}}&\gate{Z}&\gate{\mathrm{H}}  &\gate[wires=4]{e^{-\imag H T/2}} &\gate{\mathrm{H}} &\gate{Z}&\gate{\mathrm{H}}  &\gate[wires=4]{e^{-\imag H T/2}} &\qw \\
&\ket{\phi_\mathrm{in}}_2 &\qw &\gate{Z}&\qw  & &\qw&\gate{Z}&\qw & &\qw \\
&\ket{\phi_\mathrm{in}}_3 &\qw  &\gate{Z}&\qw & &\qw&\gate{Z}&\qw & &\qw \\
&\ket{\phi_\mathrm{in}}_4 &\gate{\mathrm{H}}  &\gate{Z}&\gate{\mathrm{H}} & &\gate{\mathrm{H}} &\gate{Z}&\gate{\mathrm{H}} & &\qw
\end{quantikz}
\caption{The final quantum circuit for a (4+1)-qubit example system, where the initial state has been prepared as 
$\ket{\psi_\mathrm{initial}} = \frac{1}{\sqrt{2}}(\ket{0} + \ket{1}) \ket{\phi_\mathrm{in}})$. Here the first qubit is the clock qubit, and part $B$ consists qubits 1 and 4, while part $A$ consists of qubits 2 and 3. The initial state $\ket{\psi_\mathrm{initial}}$ can be prepared by single-qubit rotations. By applying Hadamard gates before and after the globally controlled-$ZZZZ$ gate for qubits in part $B$, a controlled-$XZZX$ is implemented.  As single-qubit $Z$ commutes with $e^{-\imag H T}$, the $Z$ operations cancel out for qubits in block $A$.
}
\label{fig:final_circuit_reduced}
\end{figure*}

Our protocol features a rather efficient verification strategy, but for it to be practical, the prover must be able to prepare  $O(1)$ copies of the single-step history state of the $ZZ$-type quantum simulation. A simple approach is to run the time-independent Hamiltonian evolution generated by
\begin{equation}
    H_\mathrm{prep} = \ket{0}\bra{0} \otimes I + \ket{1}\bra{1} \otimes H,
\end{equation}
giving
\begin{equation}
    \exp(-\imag H_\mathrm{prep} T) \left[\frac{1}{\sqrt{2}} (\ket{0}+\ket{1}) \ket{\phi}\right]
    = \ket{\psi_\mathrm{hist}}.
\end{equation}
However, $H_\mathrm{prep}$ contains 3-body interaction terms. It is possible for near-term devices to implement a 3-body Hamiltonian (see for example Refs.~\cite{buchler2007three, menke2022demonstration, andrade2022engineering}),
but it may be challenging to realize $H_\mathrm{prep}$ in this way.

To circumvent the hardness of implementing 3-body interactions, we propose an echo-based method for preparing history states using only 1-qubit and 2-qubit operations.

One can easily prepare the history state of $H$ by running a half-$T$ evolution of $H$ from the state
\begin{equation}
    \ket{\psi'_\mathrm{hist}} \propto \ket{0}\exp(\tfrac{\imag}{2} H T) \ket{\phi_\mathrm{in}} + \ket{1}\exp(-\tfrac{\imag}{2} H T)\ket{\phi_\mathrm{in}}.
\end{equation}

The state $\ket{\psi'_\mathrm{hist}}$ can be prepared as follows. Since $H$ involves nearest-neighbor $ZZ$ interactions in a square lattice, one can divide all qubits into two parts such that every $ZZ$ term acts on qubits from different parts, as shown in \cref{fig:square}.
Call the filled dots part $A$, and the non-filled dots part $B$. Apply $\mathrm{CNOT}_B$ gates before and after a $T/2$ time evolution, where $\mathrm{CNOT}_B$ is controlled by the clock qubit and acts on the whole part $B$, followed by an $X$ operation (denoted by $X_0$) on the clock qubit. This gives the state (up to normalization)
\begin{equation}
\begin{aligned}
    &X_0 \mathrm{CNOT}_B\exp(-\tfrac{\imag}{2} H T)\mathrm{CNOT}_B (\ket{0}+ \ket{1})\ket{\phi_\mathrm{in}}\\
    =&\ket{1}\exp(-\tfrac{\imag}{2} H T)\ket{\phi_\mathrm{in}} + \ket{0} X_B \exp(-\tfrac{\imag}{2} H T) X_B \ket{\phi_\mathrm{in}}\\
    =&\ket{0}\exp(\tfrac{\imag}{2} H_2 T)\ket{\phi_\mathrm{in}} + \ket{1} \exp(-\tfrac{\imag}{2} H_2 T) \ket{\phi_\mathrm{in}},
\end{aligned}    
\end{equation}
where $X_B$ denotes $X$ operators acting on all qubits of part $B$.

One might be concerned that applying CNOT gates on only \emph{half} of the lattice could be difficult with a near-term device. However, one can implement $\mathrm{CNOT}_B$ using only a global controlled-$Z$ ($\mathrm{C}Z$) operator and local Hadamard operators $\mathrm{H}$. For all qubits in $B$, we perform the operation $\mathrm{H}\cdot CZ \cdot \mathrm{H}$, which is exactly a $\mathrm{CNOT}_B$ gate. For qubits in $A$, we do not apply Hadamard operators, so the controlled-$Z$ operation only adds a phase to the second state. This phase is canceled out in the end, because this effective $\mathrm{CNOT}_B$ operation is performed twice, and $Z^2 = I$.

Note that this echo approach works for more general Ising-type Hamiltonians, although they might not be easy to verify. A more general discussion can be found in \cref{supp:echo_general}. 

In summary, to realize the proposed protocol, the experimental platform should have at least $n$ system qubits and be capable of running single-qubit operations, nearest-neighbor $ZZ$ interactions, and a global $CZ$ operation, which is exactly the capability of our mostly analog + $\mathrm{GCZ}$ model of quantum computation. The quantum circuit for a 4-qubit toy model is shown in \cref{fig:final_circuit_reduced}.

\subsection{Prospects for Experimental Implementation}
\label{sec:experiment}

As explained in \cref{sec:echo_trick}, our protocol uses the mostly analog + $\mathrm{GCZ}$ capability, which roughly contains two types of ingredients: first, an analog simulator capable of implementing a $ZZ$-type Hamiltonian along with a limited number of single-qubit rotations and measurements, and, second, a global CZ gate.
The first ingredient is easily accessible in many different hardware platforms including trapped ions, neutral-atom arrays, and superconducting qubits. The second ingredient is not common in hardware architectures for digital quantum computing, but similar ideas have been explored in the context of routing and switching of single- or few-photon signals \cite{MurraySinglePhotonSwitch,RempePhotonSwitch,LiSinglePhotonSwitch} using atomic excitations, and in the case of single-photon-controlled switches \cite{VuleticAllOpticalSwitch,WaksPhotonControlledSwitch}, where a single photon can be used to switch the state of all the photons in a wave packet.

There are two possible ways for the clock qubit to globally turn simulator-qubit $Z$ gates on and off. First, coupling between the clock qubit and the simulator qubits can directly implement the $Z$ gates. Second, the clock qubit can be used to switch classical controls to the simulator qubits on and off. Both implementations are in principle possible, and each of them comes with its own unique set of hardware constraints and challenges, which we describe below.

\begin{figure}
    \centering
    \begin{tikzpicture}
    \draw[] (0,0) ellipse (1.2 and 0.2) node(bus){};
    \filldraw[black] (-0.8,0) circle (0pt) node[anchor=west]{Bus(Clock)};
    
    \filldraw (-1,-0.5) circle (3pt) node[anchor=north](q1){Qubit 1};
    \filldraw (1,-0.5) circle (3pt)node[anchor=north](q2){Qubit 2};
    \filldraw (-1,0.5) circle (3pt)node[anchor=south](q3){Qubit 3};
    \filldraw (1,0.5) circle (3pt)node[anchor=south](q4){Qubit 4};
    \filldraw (-2,0) circle (3pt)node[anchor=south](q5){Qubit 5};
    \filldraw (2,0) circle (3pt)node[anchor=south](q6){Qubits ...};

    \draw[black, thick] (q1.north) -- (-0.2,-0.2);
    \draw[black, thick] (q2.north) -- (0.2,-0.2);
    \draw[black, thick] (q3.south) -- (-0.2,0.2);
    \draw[black, thick] (q4.south) -- (0.2,0.2);
    \draw[black, thick] (q5.south) -- (-1.2,0);
    \draw[black, thick] (q6.south) -- (1.2,0);
    
    \end{tikzpicture}
    \caption{The ``bus" scheme for realizing a global $CZ$ gate. All simulation qubits are only coupled with the central ``bus" cavity mode, which behaves effectively as the clock qubit. Both the global $CZ$ gate and the $ZZ+Z$ interaction between simulation qubits can be mediated via the bus mode.}
    \label{fig:bus}
\end{figure}
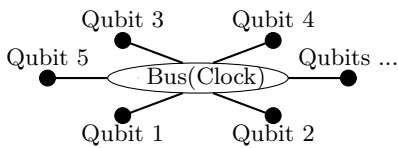

In the first method, the clock qubit itself is the signal that gives rise to $Z$ rotations of the simulator qubits. This requires the clock qubit to interact directly with \textit{all} of the simulator qubits. Therefore, this method requires the existence of a single global ``bus'' degree of freedom (e.g., a qubit or bosonic mode) that contains the two clock-qubit states $\ket{0}$ to $\ket{1}$, shown schematically in \cref{fig:bus}. The bus mode must interact strongly with each of the system qubits so that a single excitation/photon (or few photons) in the bus mode can produce a significant effect. Such interactions are possible if all of the simulator qubits are strongly coupled to a single cavity mode. Furthermore, the bus-system interaction cannot be resonant (i.e., cannot involve direct absorption of excitations in the bus); otherwise a $Z$ gate on $N$ system qubits cannot be achieved without $M \geq N$ excitations of the bus mode. The bus-system interaction must, instead, be dispersive so that the presence of excitations in the bus mode gives rise to phase shifts of simulator qubits.

While atom-cavity interactions in the single-photon strong-coupling regime are possible in atomic cavity QED, coupling strengths in the so-called strong dispersive regime which are strong enough to produce an off-resonant $CZ$ gate with only a few photons are typically only achievable with superconducting qubits coupled to microwave cavities and superconducting qubits \cite{BlaisRevModPhys}. 
Alternative implementations may also exist using a confined phonon mode, such as the vibrational modes of an ion trap \cite{monroe21}. 
Lastly, making use of strong dispersive couplings to implement a globally controlled $Z$ gate between a single excitation in the bus cavity and all of the system qubits would require building the entire simulator inside or attached to the single bus cavity. While this is certainly possible in principle, it is a bespoke feature that would need to be incorporated into the simulator as part of its initial design.

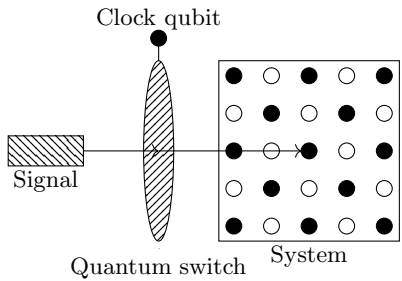
\begin{figure}
    \centering
    \begin{tikzpicture}
    \filldraw[pattern=north east lines] (-2,0) ellipse (0.2 and 1.2) node(switch){};
    \filldraw[black] (-2,-1.3) circle (0pt) node[anchor=north](switch){Quantum switch};

    \filldraw (-2,1.5) circle (3pt) node[anchor=south]{Clock qubit};
    \draw [-] (-2,1.5) -- (-2,1.2);
    
    \filldraw[pattern=north west lines] (-4,-0.2) rectangle (-3,0.2) node[anchor=north](source){};
    \node[] at (-3.5,-0.4) {Signal};

    \draw [->] (-3,0) -- (-2,0);
    \draw [->] (-2,0) -- (-0.1,0);
    \draw (-1.2,-1.2) rectangle (1.2,1.2);
    \filldraw[black] (0,0) circle (0pt) node[anchor=north]{};
    \node[] at (0,-1.4) {System};

    \filldraw (0,0) circle (3pt);
    \filldraw (0.5,0.5) circle (3pt);
    \filldraw (0.5,-0.5) circle (3pt);
    \filldraw (-0.5,0.5) circle (3pt);
    \filldraw (-0.5,-0.5) circle (3pt);
    \filldraw (1,1) circle (3pt);
    \filldraw (-1,1) circle (3pt);
    \filldraw (-1,-1) circle (3pt);
    \filldraw (1,-1) circle (3pt);
    \filldraw (0,-1) circle (3pt);
    \filldraw (0,1) circle (3pt);
    \filldraw (1,0) circle (3pt);
    \filldraw (-1,0) circle (3pt);
    
    \draw (0.5,0) circle (3pt);
    \draw (-0.5,0) circle (3pt);
    \draw (0,0.5) circle (3pt);
    \draw (0,-0.5) circle (3pt);
    \draw (1,0.5) circle (3pt);
    \draw (-1,-0.5) circle (3pt);
    \draw (0.5,1) circle (3pt);
    \draw (0.5,-1) circle (3pt);
    \draw (-0.5,1) circle (3pt);
    \draw (-0.5,-1) circle (3pt);
    \draw (-1,0.5) circle (3pt);
    \draw (1,-0.5) circle (3pt);
    
    \end{tikzpicture}
    \caption{The quantum switch scheme. Here the simulation qubits are assigned in the square lattice as usual. A photon source gives signals that implement $Z$ operations for each simulation qubit. A high-performance quantum switch, controlled by the clock qubit which could be in superposition, determines whether the signal can be received by simulation qubits or not, which realizes a global $CZ$ gate.}
    \label{fig:switch}
\end{figure}

Because of the stringent hardware constraints for the ``bus'' method of implementing the clock qubit, it is worth considering other methods in which the operation of the clock qubit is more separated from the operation of the quantum simulator under test. Separating these two means that instead of being used to switch the simulator qubits directly, the clock qubit must now switch the \textit{control} signals for single-qubit $Z$ gates on and off. This architecture, shown schematically in \cref{fig:switch}, provides significantly more separation between the design constraints of the simulator and those of the clock qubit, but it requires the clock qubit to control a very high-performance quantum switch. In particular, it is not sufficient to use a classical switch with an extremely low switching energy provided by the clock qubit; instead, the switch itself must be able to exist in a superposition between on and off. Such a superposition switching state is extremely challenging to achieve with large control signals since there are many opportunities to lose a photon (and thus destroy the superposition). Compared to other architectures, superconducting qubits typically require very low switching power---as low as a few photons---so they are a likely candidate for implementation of the necessary quantum switch. For example, a broadband and high-dynamic-range switch such as the one demonstrated in \textcite{WallraffSuperconductingSwitch} could be converted to use, e.g., a galvanically coupled fluxonium qubit \cite{Manucharyan:2009vf} as the switching element. In the optical domain, single-photon controlled switches have been implemented using atomic ensembles \cite{VuleticAllOpticalSwitch} and self-assembled semiconductor quantum dots \cite{WaksPhotonControlledSwitch} as the switching medium.

\section{Summary \& Discussions}
\label{sec:summary}
In summary, we have proposed a novel scheme for demonstrating quantum computational ability based on verification of analog quantum simulation. The verifier in the scheme need only be capable of polynomial-time classical computation.
The prover can be an analog quantum simulator with the additional power of single-qubit operations and a specific global $CZ$ gate, and only needs to be able to prepare a constant number of samples, independent of the system size. Additionally, we assume the prover can perform trusted measurements. We also described some possible near-term experimental implementations of the global $CZ$ gate. 

\textcite{hangleiter_direct_2017} propose another certification scheme that was applied in \cite{bermejo2018architectures} to verify measurement outcomes using only local measurements. The method in \cite{hangleiter_direct_2017} can even verify $\mathsf{BQP}$-complete computation encoded through the Feynman-Kitaev mapping, but it requires $O(n^2)$ samples of the output state for the $ZZ+Z$ Hamiltonian evolution, which is more expensive than our constant-sample-complexity scheme. Our improvement is achieved by a combination of the single-step Feynman-Kitaev encoding and the commuting nature of the $ZZ+Z$ Hamiltonian (or the $ZZ$ Hamiltonian when single-qubit $Z$s are absorbed). In fact, our protocol can verify all commuting Hamiltonians with constant sample complexity if entangled multi-qubit measurements are allowed, but it is unclear whether there are also near-term honest-prover strategies in this more general case. We discuss this in more detail in \cref{supp:echo_general}.

It is worth noting that there can be a tradeoff between the verification cost and the difficulty of experimental realization. Our verification protocol presented in \cref{sec:qaprotocol} does not rely on the condition that in $H$, the coefficient of every term $Z_i Z_j$ is the same ($\pi/4$), but this uniformity makes it possible to simulate the system using $2^{O(\sqrt{n})}$-time classical computation, so that $n=O(\lambda^2)$. If instead the coefficients are randomly selected, then the above simulation is no longer available, and we can conjecture the classical simulation cost to be $2^{\Omega(n)}$, as in \cite{dalzell2020many}. In this case, the number of qubits, the number of single-qubit measurements, and the classical computational cost can all be reduced to $O(\lambda)$---at the cost of more difficult history state preparation---since non-uniform Hamiltonian evolution is in general more challenging.

As the main technical tool of this work, we studied a simplified single-step Feynman-Kitaev construction and developed a scheme to lower bound the output fidelity $F_\mathrm{output}$ (and subsequently the TVD between ideal and experimental distributions) using three parameters. In fact, the lower bound holds for any unitary $U$, but the three parameters may not be efficiently estimatable in general. One might ask if we can simply combine the protocol of \textcite{fitzsimons2018post} with our single-step construction to verify arbitrary quantum operations, such as non-commuting Hamiltonian evolutions or digital quantum circuits. We do not have a definite answer, but this seems difficult for most hard-to-simulate unitaries because they generally decompose into exponentially many Pauli terms and, unlike $ZZ+Z$ or $ZZ$ Hamiltonian evolution, their \emph{de facto} measurement outcome cannot be efficiently deduced from $\mathrm{poly}(\lambda)$ single-qubit measurements. 

Experimental implementation of the protocol would be of significant interest. Although it might be difficult to implement quantum communication in the adversarial scenario, our protocol could be a useful tool for experimentalists to benchmark the quality of their devices, because the quality of initial state preparation and that of Hamiltonian evolution can be estimated separately and precisely. As shown in \cref{supp:relate_yikai}, if the noise pattern is known to be fully stochastic instead of coherent, the experimentalist only needs to achieve output fidelity $0.708$, which is significantly easier than the bound of $0.915$ in the fully coherent case.

Finally, our approach may have applications to realizing near-term quantum advantage even in devices capable of digital quantum computation. Reconfigurable atom arrays \cite{beugnon07,bluvstein2023logical} may be one such system.
In these arrays, physical qubits (realized by individual neutral atoms controlled by optical tweezers) can be moved accurately on the 2-D plane in parallel, and transversal $CZ$ gates are available. Therefore, our global $CZ$ gate can be implemented as follows. One can first prepare a large $n$-qubit GHZ state that behaves as the clock qubit. The GHZ state preparation can be implemented by either performing a sequence of CNOT gates or using constant-depth unitary operations interleaved with measurements and classical computations \cite{watts2019exponential}. One can then can move all qubits in the GHZ state such that every system qubit pairs with a GHZ qubit. Next, using the Levine-Pichler gate \cite{levine2019parallel}, $CZ$ gates can be implemented in parallel for every pair of system and GHZ qubits, effectively implementing the global $CZ$ acting on all system qubits. There is also a multi-step solution to mitigate the hardness of GHZ preparation: since our system is a $\sqrt{n}\times \sqrt{n}$ square lattice, it suffices to prepare a 1-D $\sqrt{n}$-qubit GHZ state, and apply the transversal $CZ$ gate  $\sqrt{n}$ times to achieve the same global $CZ$ gate. This proposal is depicted in \cref{fig:reconfig}.

While
\emph{digital} reconfigurable atom arrays are capable of even more powerful quantum operations than the mostly-analog + $\mathrm{GCZ}$ commuting model,
it may still be worth performing our proposed experiment using Rydberg atoms. Running our verification protocol gives several quantitative performance measures ($F_\mathrm{in}$ and $|\langle O_{10} \rangle|^2$),
and can thus be used to benchmark the performance of this fast-developing platform in a sample-efficient manner.

\begin{figure}
    \centering
    \begin{tikzpicture}

    \filldraw[gray] (-1.2,1.2) circle (3pt);
    \filldraw[gray] (-0.7,1.2) circle (3pt);
    \filldraw[gray] (-0.2,1.2) circle (3pt);
    \filldraw[gray] (0.3,1.2) circle (3pt);
    \filldraw[gray] (0.8,1.2) circle (3pt);
    
    \draw[rounded corners] (-1.35, 0.8) rectangle (-0.86, 1.35) {};
    \draw[rounded corners] (-0.85, 0.8) rectangle (-0.36, 1.35) {};
    \draw[rounded corners] (-0.35, 0.8) rectangle (0.14, 1.35) {};
    \draw[rounded corners] (0.15, 0.8) rectangle (0.64, 1.35) {};
    \draw[rounded corners] (0.65, 0.8) rectangle (1.14, 1.35) {};
    
    \filldraw (0,0) circle (3pt);
    \filldraw (0.5,0.5) circle (3pt);
    \filldraw (0.5,-0.5) circle (3pt);
    \filldraw (-0.5,0.5) circle (3pt);
    \filldraw (-0.5,-0.5) circle (3pt);
    \filldraw (1,1) circle (3pt);
    \filldraw (-1,1) circle (3pt);
    \filldraw (-1,-1) circle (3pt);
    \filldraw (1,-1) circle (3pt);
    \filldraw (0,-1) circle (3pt);
    \filldraw (0,1) circle (3pt);
    \filldraw (1,0) circle (3pt);
    \filldraw (-1,0) circle (3pt);
    
    \draw (0.5,0) circle (3pt);
    \draw (-0.5,0) circle (3pt);
    \draw (0,0.5) circle (3pt);
    \draw (0,-0.5) circle (3pt);
    \draw (1,0.5) circle (3pt);
    \draw (-1,-0.5) circle (3pt);
    \draw (0.5,1) circle (3pt);
    \draw (0.5,-1) circle (3pt);
    \draw (-0.5,1) circle (3pt);
    \draw (-0.5,-1) circle (3pt);
    \draw (-1,0.5) circle (3pt);
    \draw (1,-0.5) circle (3pt);
    \end{tikzpicture}
    \caption{The reconfigurable atom array scheme. The gray dots represent qubits in a $\sqrt{n}$-qubit GHZ state. Local $CZ$ gates can be realized between pairs of GHZ qubits and system qubits in parallel. Then the GHZ qubits are moved down to the next row and the parallel $CZ$ gates are repeated.}
    \label{fig:reconfig}
\end{figure}
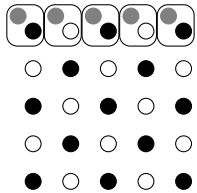

\section*{Acknowledgments}
We thank Wentai Deng and Ruozhen Gong for helpful discussions.
This work received support from the National Science Foundation (QLCI grant OMA-2120757).
Z.L. acknowledges financial support by the QuICS Lanczos Graduate Fellowship. D.D.\ acknowledges support by the NSF GRFP under Grant No.~DGE-1840340, an LPS Quantum Graduate Fellowship, and the U.S. Department of Energy, Office of Science, Office of Advanced Scientific Computing Research, Quantum Testbed Pathfinder program (award number DE-SC0019040).
A.V.G.\ was also supported in part by AFOSR, DoE ASCR Accelerated Research in Quantum Computing program (award No.~DE-SC0020312), DoE ASCR Quantum Testbed Pathfinder program (awards No.~DE-SC0019040 and No.~DE-SC0024220),  NSF STAQ program,  AFOSR MURI, and DARPA SAVaNT ADVENT. Support is also acknowledged from the U.S.~Department of Energy, Office of Science, National Quantum Information Science Research Centers, Quantum Systems Accelerator. 
D.H. acknowledges funding from the US department of defense through a QuICS Hartree fellowship. 
Y.K.L.\ acknowledges support from NIST, and from AFOSR MURI Scalable Certification of Quantum Computing Devices and Networks.

\bibliographystyle{unsrtnat}
\bibliography{mainbib}

\clearpage

\appendix

\renewcommand{\thesection}{S\arabic{section}}
\renewcommand{\thesubsection}{\arabic{subsection}}
\renewcommand{\theequation}{\thesection.\arabic{equation}}

\makeatletter
\renewcommand*{\p@subsection}{\thesection.}
\makeatother

\onecolumngrid

\section*{Appendices}
In these appendices,
we present details omitted from the main text.
In \cref{supp:conjectures}, we state and briefly explain the conjectures used to establish computational hardness \cite{bermejo2018architectures,ringbauer2023verifiable}.
In \cref{supp:relate_yikai}, we lower bound the output fidelity and the total variation distance between distributions using the parameters in our verification scheme.
In \cref{supp:noisy}, we discuss noisy measurements and estimate the noise rate that both verification and sampling can tolerate.
In \cref{supp:martingale}, we discuss an additional soundness property of our protocol against correlated output states using martingale inequalities.
In \cref{supp:echo_general}, we generalize the echo method presented in the main text to more general Ising-type Hamiltonians.

\section{Conjectures for the classical hardness}
\label{supp:conjectures}

As mentioned in the main text, the classical hardness of $X$-basis sampling from a state produced by $(ZZ+Z)$-Hamiltonian evolution is proven in \textcite{bermejo2018architectures} and \textcite{ringbauer2023verifiable} under several plausible conjectures, which we review in this section.

\begin{conjecture}[Polynomial Hierarchy---Conjecture 1 in \cite{bermejo2018architectures}]
\label{conj1}
The polynomial hierarchy is infinite.
\end{conjecture}

The second conjecture considers the hardness of a random nearest-neighbor Ising model on an $n \times m$
square lattice where $m$ grows at least linearly with $n$, with the Hamiltonian
\begin{equation}
    H^{(\alpha, \beta)}=\sum_{i,j}\frac{\pi}{4}Z_i Z_j - \sum_i h_i^{(\alpha,\beta)}Z_i,
\end{equation}
where $h_i^{(\alpha,\beta)} = \frac{\pi}{4}-\frac{\alpha_i + \beta_i}{2}$ with $\alpha_i \in \{0,\pi\}, \beta_i \in \{0,\pi/4\}$ chosen uniformly at random.

\begin{conjecture}[Average-case complexity---Conjecture 2 in \cite{bermejo2018architectures} and conjecture in \cite{ringbauer2023verifiable}]
\label{conj2}
Let $Z^{(\alpha,\beta)} := \Tr\bigl(\smash{e^{\imag H^{(\alpha, \beta)}}}\bigr)$. Approximating $|Z^{(\alpha,\beta)}|^2$ up to relative error $1/4 + o(1)$ for any 0.001 fraction of the field configurations is \#P-hard.
\end{conjecture}

The last conjecture is about anti-concentration of the output distribution. Consider a one-dimensional nearest-neighbor $n$-qubit $\Theta(n)$-depth random circuit
\begin{equation}
    \mathcal{C} = \left[\prod_{i=1}^{n-1}CZ_{i,i+1}\right] \left[\prod_{i=1}^n Z_i^{c_i} e^{-\imag \frac{\pi}{4} d_i Z_i }H_i\right],
\end{equation}
where $c_i,d_i$ are uniformly randomly chosen from $\{0,1\}$ and $H_i$ are Hadamard gates.

\begin{conjecture}[Anti-concentration---Conjecture 3 in \cite{bermejo2018architectures}]
\label{conj3}
For the random circuit $\mathcal{C}$ described above,
\begin{equation}
    \Pr_\mathcal{C} \left( |\langle x|\mathcal{C}|0\rangle^{\otimes n}|^2 \geq \frac{1}{2^n} \right) \geq \frac{1}{e}
\end{equation}
for any binary string $x \in \{0,1\}^n$.
\end{conjecture}

\section{Relating the parameters to the total variation distance}
\label{supp:relate_yikai}

In this appendix, we derive an upper bound on the total variation distance of interest, $\TVD(P_\mathrm{ideal}, P_\mathrm{real})$, in terms of the parameters $F_\mathrm{in}$, $|\langle O_{10}\rangle|^2$, and $p_\mathrm{samp}$. We use the same definition of $\rho$ and $\ket{\psi_i}$ as in \cref{eqn-rho-diag,eqn-psi-diag}.

First, we relate the TVD and the output fidelity
\begin{equation}
    F_\mathrm{out}(\rho) := \frac{ \sum_i p_i |\beta_i|^2 \left| \langle \phi'_i|U|\phi_\mathrm{in} \rangle \right|^2} {\sum_i p_i |\beta_i|^2}.
\end{equation}
This is the fidelity between the state used for sampling, $\rho$, and $U\ket{\phi_\mathrm{in}}$, since the state corresponding to the ``output" of the computation is $\ket{\phi'_i}$ for all $i$.

In the second step, we derive a lower bound on the state fidelity in terms of the parameters. We lower bound $F_\mathrm{out}(\rho)$ using only the parameters $\Tr[\rho O_{10}]$ and $F_\mathrm{in}(\rho)$. We find
\begin{equation}
    F_\mathrm{out}(\rho) \geq 16 |{\Tr[\rho O_{10}]}|^2 + 3F_\mathrm{in}(\rho) - 6
\end{equation}
up to higher-order terms.
As a sanity check, if the history state is perfectly prepared, both $|{\Tr[\rho O_{10}]}|^2$ and $F_\mathrm{in}$ should take their maximum values, which are $1/4$ and $1$ (as shown later in this section), giving $F_\mathrm{out}=1$ as expected.

\subsection{Relating the total variation distance to the output fidelity}

To demonstrate quantum advantage, we generate samples from the desired distribution $P_\mathrm{ideal}$ defined by $U\ket{\phi_\mathrm{in}}$ with total variation distance (TVD) less than $\delta = 0.292$ as per \textcite{ringbauer2023verifiable}. 
Therefore, we would like to relate the fidelity $F_\mathrm{out}$ obtained from the measurements to the distance between the distribution $P_\mathrm{real}$ corresponding to the classical mixture of $\ket{\phi'_i}$s (i.e., $\sum_i p_i |\beta_i|^2 \ket{\phi'_i}\bra{\phi'_i}$) and the ideal distribution $P_\mathrm{ideal}$.

Let $\|\cdot\|_{\Tr}$ be the trace norm (Schatten 1-norm). The TVD between probability distributions generated by measurements on quantum states is upper bounded by the trace distance between those states, which is in turn related to the fidelity:
\begin{equation}\label{eq:tvd-fidelity}
\begin{aligned}
    \TVD(P_\mathrm{ideal}, P_\mathrm{real})
    &\leq \frac{1}{2}  \left\|U\ket{\phi_\mathrm{in}}\bra{\phi_\mathrm{in}}U^\dagger - \sum_i p_i |\beta_i|^2 \ket{\phi'_i}\bra{\phi'_i} \right\|_\mathrm{Tr}\\
    &\leq \sqrt{1-F\left(U\ket{\phi_\mathrm{in}}\bra{\phi_\mathrm{in}}U^\dagger, \sum_i p_i |\beta_i|^2 \ket{\phi'_i}\bra{\phi'_i}\right)}\\
    &=\sqrt{1-\sum_i p_i |\beta_i|^2 F(\ket{\phi'_i},U\ket{\phi_\mathrm{in}})} = \sqrt{1-F_\mathrm{out}}.
\end{aligned}
\end{equation}
Thus, $\TVD(P_\mathrm{ideal}, P_\mathrm{real}) \leq 0.292$ is satisfied if
\begin{equation}
   \label{eq:actualnumber} F_\mathrm{out} \geq 0.915 > 1 -\delta^2,
\end{equation}
where we use $\delta = 0.292$.

We also observe that the output fidelity requirement can be relaxed to $0.708$ if the noise in the system is known to be fully stochastic. We discuss this in \cref{subsec:relax_yikai}.

\subsection{Lower bounding the output fidelity using the parameters}

As a mathematical tool, we define an inner product based on the (not explicitly known) diagonalization of $\rho$. Suppose $\rho = \sum_{i=1}^{2^{n+1}} p_i \ket{\psi_i}\bra{\psi_i}$ and there exists an integer $N_{\neq 0}>1$ such that $p_i > 0$ for all $1\leq i\leq N_{\neq 0}$ and $p_i = 0$ for all $N_{\neq 0}<i\leq 2^{n+1}$. The inner product $\langle\cdot,\cdot\rangle_\rho$ is defined for the $N_{\neq 0}$-dimensional complex vector space $V=\mathbb{C}^{N_{\neq 0}}$ as
\begin{equation}
    \langle \vec{A}, \vec{B} \rangle_\rho :=\sum_{1\leq i\leq N_{\neq 0}} p_i A_i B^*_i,
\end{equation}
where $\vec{A}:=(A_1,A_2,\dots,A_{N_{\neq 0}})^\mathrm{T}$ and $\vec{B}:=(B_1,B_2,\dots,B_{N_{\neq 0}})^\mathrm{T}$ are vectors in $V$. It is straightforward to verify that for any valid density matrix $\rho$, the vector space $V$ equipped with $\langle\cdot,\cdot \rangle_\rho$ is an inner product space. Therefore, one can define the norm of a vector in $V$ as
\begin{equation}
\|\vec{A}\|^2 := \langle  \vec{A},\vec{A} \rangle_\rho = \sum_i p_i |A_i|^2.
\end{equation}

Next, we define several vectors to help represent the state and the parameters: the \emph{input fidelity vector} $\vec{f}_\mathrm{in}$, the \emph{propagation fidelity vector} $\vec{f}_\mathrm{in}$, the \emph{output fidelity vector} $\vec{f}_\mathrm{out}$, the \emph{$\alpha$ coefficient vector} $\vec{\alpha}$, the \emph{$\beta$ coefficient vector} $\vec{\beta}$, and the \emph{$\gamma$ coefficient vector} $\vec{\gamma}$ for a given mixed state $\rho$, namely
\begin{equation}
    \begin{aligned}
        \vec{f}_\mathrm{in} &:=(\dots,\langle\phi_i|\phi_\mathrm{in}\rangle,\dots)^\mathrm{T},\\
        \vec{f}_\mathrm{prop} &:=(\dots,\langle\phi'_i|U|\phi_i\rangle,\dots)^\mathrm{T},\\
        \vec{f}_\mathrm{out} &:= (\dots,\langle\phi'_i|U|\phi_\mathrm{in}\rangle,\dots)^\mathrm{T},\\
        \vec{\alpha} &:=(\dots,\alpha_i,\dots)^\mathrm{T},\\
        \vec{\beta} &:=(\dots,\beta,\dots)^\mathrm{T},\\
        \vec{\gamma} &:=(\dots, \alpha_i \beta^*_i,\dots)^\mathrm{T},
    \end{aligned}
\end{equation}
respectively.
Note that $\|\vec{\gamma}\|^2 \leq 1/4$ and $\|\vec{f}_\mathrm{in}\|^2,\|\vec{f}_\mathrm{prop}\|^2,\|\vec{f}_\mathrm{out}\|^2\leq 1$ since $|\alpha_i|^2 + |\beta_i|^2 = 1$, $\sum_i p_i = 1$, and fidelities are at most 1.

Observe that $p_\mathrm{samp}$ is the same as $\| \vec{\alpha} \|^2$. Another parameter, $\Tr[\rho O_{10}]$, can be written as the inner product of two of the above vectors:
\begin{equation}
\begin{aligned}
    \Tr[\rho O_{10}] &= \sum_i p_i \alpha_i \beta^*_i \langle\phi'_i | U | \phi_i \rangle  = \langle \vec{\gamma}, \vec{f}_\mathrm{prop}\rangle_\rho.
\end{aligned}
\end{equation}

Using the Cauchy-Schwarz inequality, we find
\begin{equation}
    |{\Tr[\rho O_{10}]}|^2 = |\langle \vec{\gamma},\vec{f}_\mathrm{prop} \rangle|^2 \leq \|\vec{\gamma}\|^2 \|\vec{f}_\mathrm{prop}\|^2 \leq 1/4.
\end{equation}
Since $\|\vec{\gamma}\|^2 \leq 1/4$ and $\|\vec{f}_\mathrm{prop}\|^2\leq 1$, the above inequality implies that
\begin{equation}
    \begin{aligned}
        4\|\vec{\gamma}\|^2 &\geq |{\Tr[\rho O_{10}]}|^2,\\
        \|\vec{f}_\mathrm{prop}\|^2 &\geq 4|{\Tr[\rho O_{10}]}|^2.
    \end{aligned}
\end{equation}
If the prover performs well, then the estimated value $\Tr[\rho O_{10}]$ should be close to $1/4$, $\|\vec{\alpha}\|^2$ should be close to $1/2$, and $F_\mathrm{in}$ should be close to $1$. Therefore we write them as $\Tr[\rho O_{10}] = 1/4 - \epsilon$, $\|\vec{\alpha}\|^2 = 1/2 + \epsilon' =  1-\|\vec{\beta}\|^2$, and $F_\mathrm{in}=1-\epsilon''$, where $\epsilon, \epsilon', \epsilon''$ are all small and $\epsilon,\epsilon'' > 0$. This also implies that $\|\vec{\gamma}\|^2 = \sum_i p_i |\alpha_i|^2 |\beta_i|^2 \geq 1/4-\epsilon$ and $\|\vec{f}_\mathrm{prop}\|^2 \geq 1-4\epsilon$. 

Recall that our final objective is to lower bound $F_\mathrm{out}(\rho)$. We start by giving a lower bound on $\|\vec{f}_\mathrm{in}\|^2$ in terms of $F_\mathrm{in}$.

First, the Cauchy-Schwarz inequality gives
\[
\begin{split}
F_\mathrm{in}(\rho)
&= \frac{1}{\|\vec{\alpha}\|^2} \sum_i p_i |\alpha_i|^2 \left|\braket{\phi_i}{\phi_\mathrm{in}}\right|^2 \\
&\leq \frac{1}{\|\vec{\alpha}\|^2} \sum_i p_i |\alpha_i|^2 \left|\braket{\phi_i}{\phi_\mathrm{in}}\right| \\
&\leq \frac{1}{\|\vec{\alpha}\|^2} \left(\sum_i p_i |\alpha_i|^4 \right)^{1/2} \|\vec{f}_\mathrm{in}\|.
\end{split}
\]
Plugging in the identity $|\alpha_i|^4 = |\alpha_i|^2 - |\alpha_i|^2 |\beta_i|^2$, we get
\[
F_\mathrm{in}(\rho)
\leq \frac{1}{\|\vec{\alpha}\|^2} (\|\vec{\alpha}\|^2 - \|\vec{\gamma}\|^2)^{1/2} \|\vec{f}_\mathrm{in}\|.
\]
As before, suppose that $\|\vec{\gamma}\|^2 = 1/4-\epsilon$ and $\|\vec{\alpha}\|^2 = 1/2 + \epsilon'$. This implies that
\[
\begin{split}
F_\mathrm{in}(\rho)
&\leq \frac{1}{\tfrac{1}{2} + \epsilon'} \left(\tfrac{1}{2} + \epsilon' - \tfrac{1}{4} + \epsilon \right)^{1/2} \|\vec{f}_\mathrm{in}\|.
\end{split}
\]
We can rewrite this as
\[
\|\vec{f}_\mathrm{in}\| \geq \frac{\tfrac{1}{2} + \epsilon'}{\sqrt{\tfrac{1}{4}+\epsilon' + \epsilon}} F_\mathrm{in} = (1-2\epsilon) F_\mathrm{in} + O(\epsilon'^2)+O(\epsilon^2)+O(\epsilon\epsilon').
\]

Next, since $|\langle \phi'_i | U | \phi_\mathrm{in} \rangle|^2 \geq \left| \bra{\phi'_i}U\ket{\phi_i} \right|^2 \left| \braket{\phi_i}{\phi_\mathrm{in}} \right|^2$, we have
\[
F_\mathrm{out}(\rho)
\geq \frac{1}{\|\vec{\beta}\|^2} \sum_i p_i |\beta_i|^2 \left| \bra{\phi'_i}U\ket{\phi_i} \right|^2 \left| \braket{\phi_i}{\phi_\mathrm{in}} \right|^2.
\]
Note that for any $\delta_1, \delta_2 \in [0,1]$, we have $(1-\delta_1) (1-\delta_2) \geq 1 - \delta_1 - \delta_2 = (1-\delta_1) + (1-\delta_2) - 1$. Using this inequality, we can write
\[
\begin{split}
F_\mathrm{out}(\rho)
&\geq \frac{1}{\|\vec{\beta}\|^2} \sum_i p_i |\beta_i|^2 \left( \left| \bra{\phi'_i}U\ket{\phi_i} \right|^2 + \left| \braket{\phi_i}{\phi_\mathrm{in}} \right|^2 - 1 \right) \\
&= -1 + \frac{1}{\|\vec{\beta}\|^2} \sum_i p_i (1-|\alpha_i|^2) \left( \left| \bra{\phi'_i}U\ket{\phi_i} \right|^2 + \left| \braket{\phi_i}{\phi_\mathrm{in}} \right|^2 \right) \\
&= -1 + \frac{1}{\|\vec{\beta}\|^2} \left( \|\vec{f}_\mathrm{prop}\|^2 + \|\vec{f}_\mathrm{in}\|^2 \right) - \frac{1}{\|\vec{\beta}\|^2} \sum_i p_i |\alpha_i|^2 \left| \bra{\phi'_i}U\ket{\phi_i} \right|^2 - \frac{\|\vec{\alpha}\|^2}{\|\vec{\beta}\|^2} F_\mathrm{in}(\rho).
\end{split}
\]
The second-to-last term can be bounded in terms of $\|\vec{f}_\mathrm{prop}\|$, using the same argument we used to relate $F_\mathrm{in}(\rho)$ and $\|\vec{f}_\mathrm{in}\|$. This yields
\begin{equation}
    \begin{aligned}
\frac{1}{\|\vec{\beta}\|^2} \sum_i p_i |\alpha_i|^2 \left| \bra{\phi'_i}U\ket{\phi_i} \right|^2 
\leq \frac{\sqrt{\tfrac{1}{2} + \epsilon' -\tfrac{1}{4} +\epsilon}}{\tfrac{1}{2}-\epsilon'} \|\vec{f}_\mathrm{prop}\|
= (1+4\epsilon'+2\epsilon) \|\vec{f}_\mathrm{prop}\| + O(\epsilon'^2) + O(\epsilon^2) + O(\epsilon\epsilon').
    \end{aligned}
\end{equation}
Plugging this into the preceding equation, we get
\begin{equation}
    \begin{aligned}
F_\mathrm{out}(\rho)
&\geq -1 + \frac{1}{\|\vec{\beta}\|^2} \left( \|\vec{f}_\mathrm{prop}\|^2 + \|\vec{f}_\mathrm{in}\|^2 \right) - \frac{\sqrt{\tfrac{1}{2} + \epsilon' -\tfrac{1}{4} +\epsilon}}{\tfrac{1}{2}-\epsilon'} \|\vec{f}_\mathrm{prop}\| - \frac{\|\vec{\alpha}\|^2}{\|\vec{\beta}\|^2} F_\mathrm{in}(\rho) \\
&\geq -1 + \frac{2}{1-2\epsilon'} \left( \|\vec{f}_\mathrm{prop}\|^2 + \|\vec{f}_\mathrm{in}\|^2 \right) - \frac{\sqrt{\tfrac{1}{2} + \epsilon' -\tfrac{1}{4} +\epsilon}}{\tfrac{1}{2}-\epsilon'} \|\vec{f}_\mathrm{prop}\| - \frac{\tfrac{1}{2}+\epsilon'}{\tfrac{1}{2}-\epsilon'}F_\mathrm{in}(\rho)\\
&= 1-16\epsilon - 3\epsilon'' + \mathrm{h.o.},
    \end{aligned}
\end{equation}
where $\mathrm{h.o.}$ indicates higher-order terms in $\epsilon,\epsilon',\epsilon''$.
Numerically, this first-order approximation of the lower bound has absolute error at the $10^{-3}$ order of magnitude if all of $\epsilon,|\epsilon'|,\epsilon''$ are upper bounded by $0.02$. We have thus established \cref{thm:lowerbounddistance}.

\subsection{Relaxing the fidelity requirement for fully stochastic noise models}
\label{subsec:relax_yikai}

We notice that inequality~\eqref{eq:tvd-fidelity} can be improved to get a bound that approaches
\begin{equation}
\label{eq:tvd-mixed}
    \TVD(P_\mathrm{ideal}, P_\mathrm{real}) \leq 1-F_\mathrm{out}
\end{equation}
in cases where the errors are stochastic rather than coherent. Let $\rho_\mathrm{real}:=\sum_i p_i |\beta_i|^2 \ket{\phi'_i}\bra{\phi'_i}$ be the real state (that is, the state prepared in the experiment), and let $\sigma = \ket{\psi}\bra{\psi}$ be the ideal pure state. The real state $\rho_\mathrm{real}$ has fidelity $F_\mathrm{out} = \bra{\psi}\rho_\mathrm{real}\ket{\psi} = 1-\delta_f$ (where $\delta_f$ is the ``infidelity''). 

Furthermore, assume that $\rho_\mathrm{real}$ is mixed, in the sense that $\Tr(\rho_\mathrm{real}^2) = 1-\delta_p$ (where $\delta_p$ is the ``impurity''). This assumption can be checked by estimating $\Tr(\rho_\mathrm{real}^2)$ using either randomized measurements \cite{elben2023randomized} or the swap test. (The former method is appropriate for small quantum systems where the experimenter has a relatively limited degree of control; the latter method is capable of handling much larger quantum systems, but requires more sophisticated quantum control.)

Define projectors $\Pi_0 := \ket{\psi}\bra{\psi}$ and $\Pi_1 := I-\Pi_0$. Write the state in block-diagonal form as $\rho_\mathrm{real} = \rho_{00} + \rho_{01} + \rho_{10} + \rho_{11}$, where $\rho_{ab} := \Pi_a\rho_\mathrm{real}\Pi_b$ for $a,b \in \lbrace 0,1 \rbrace$.

Let $\|\cdot\|_F$ be the Frobenius norm (i.e., the Schatten 2-norm). Then we can upper bound the trace distance between $\rho_\mathrm{real}$ and $\sigma$ as follows:
\begin{equation}
\begin{split}
\|\rho_\mathrm{real}-\sigma\|_{\Tr} &\leq \|\rho_{00}-\sigma\|_{\Tr} + \|\rho_{11}\|_{\Tr} + \|\rho_{01}\|_{\Tr} + \|\rho_{10}\|_{\Tr} \\
&= 2\delta_f + 2\|\rho_{01}\|_{\Tr}.
\end{split}
\end{equation}
We have
\begin{equation}
\begin{split}
\|\rho_{01}\|_{\Tr} &= \|\rho_{01}\|_F \\
&= \tfrac{1}{\sqrt{2}} (\Tr(\rho_\mathrm{real}^2) - \|\rho_{00}\|_F^2 - \|\rho_{11}\|_F^2)^{1/2} \\
&\leq \tfrac{1}{\sqrt{2}} (\Tr(\rho_\mathrm{real}^2) - \|\rho_{00}\|_F^2)^{1/2} \\
&= \tfrac{1}{\sqrt{2}} (1 - \delta_p - (1-\delta_f)^2)^{1/2} \\
&= \tfrac{1}{\sqrt{2}} (2\delta_f - \delta_f^2 - \delta_p)^{1/2}.
\end{split}
\end{equation}
Therefore,
\begin{equation}
\frac{1}{2} \|\rho_\mathrm{real}-\sigma\|_{\Tr} \leq \delta_f + \sqrt{\delta_f - \delta_f^2/2 - \delta_p/2}.
\end{equation}

This bound can be compared to inequalities~\eqref{eq:tvd-fidelity} and~\eqref{eq:tvd-mixed}. When $\rho$ is a pure state, we have $\delta_p = 0$, so the above bound is roughly $\sqrt{\delta_f}$, which looks like inequality~\eqref{eq:tvd-fidelity}. When $\rho$ is highly mixed, $\delta_p$ can be as large as $\delta_p \approx 2\delta_f - \delta_f^2$, so the above bound is roughly $\delta_f$, which looks like inequality~\eqref{eq:tvd-mixed}. This implies that, when the noise model is known to be fully stochastic, the output state fidelity need only be at least $1-\delta = 0.708$ to demonstrate quantum advantage, according to inequality~\eqref{eq:tvd-mixed}.

\section{Noisy Measurements}
\label{supp:noisy}

In the analysis in the main article, we assume that all measurements are perfect. In this appendix, we discuss the potential negative effects of noisy measurements in both verification and sampling. We also show that the tolerable noise rate in measurements for an $n$-qubit system is $\epsilon \ll 1/n$.

\subsection{Noisy measurements in verification}
Let us first discuss the estimation of $|\langle O_{10}\rangle|^2 = |\langle X\otimes U \rangle + \imag\langle Y\otimes U\rangle|^2$. When $\epsilon \ll 1/n$, the number of erroneous measurements in each estimation of the \emph{de facto} value of $X\otimes U$ or $Y\otimes U$ is much less than 1. Therefore, the mean values measured for both quantities only deviate by up to $n\epsilon \langle X\otimes U \rangle$ and $n\epsilon \langle Y\otimes U \rangle$ due to the measurement errors, leading to constant-factor errors in the estimation of $|\langle O_{10}\rangle|^2$. Hence, the error rate must be sufficiently small, e.g., $\epsilon = \frac{1}{100n}$, such that the estimated value can still be in the range of acceptance.

Similarly, we require the measurement error to be as small as $\frac{1}{100n}$ to estimate $F_\mathrm{in}$ to sufficiently high precision, because the value of $N_\mathrm{in+0}$ could be lowered by $N_\mathrm{M} n \epsilon$ when measuring $N_\mathrm{M}$ samples. This may lead to a constant-factor error (of order $n\epsilon$) in $F_\mathrm{in,M}$.

\subsection{Noisy measurements in sampling}

In the following lemma, we show that we can still sample from a classically intractable distribution if the measurement error is much lower than $1/n$.

\begin{lemma}
    If $F_\mathrm{output}  = 1-\delta_f$, and all measurements have the same error rate $\epsilon \ll 1/n$, then the measurement outcomes sample from a distribution $P_\mathrm{real}$ with $\TVD(P_\mathrm{real}, P_\mathrm{ideal}) \leq \delta' = \sqrt{\delta_f} + O(1)$.
\end{lemma}

\begin{proof}
    Since there are $n$ Hadamard measurements to be performed, the probability of having no error in the measurements is
    \begin{equation}
        p_\mathrm{measure} = (1-\epsilon)^n \approx 1-\epsilon n.
    \end{equation}
    Therefore, there is a $1-\epsilon n$ probability that the measurement outcome samples from a distribution that is $\sqrt{\delta_f}$ away from the ideal distribution in terms of TVD. In the worst case, we simply assume the distribution of errorneous measurements has maximum TVD from the ideal distribution, which is 1. Hence, the TVD between the real experiment distribution and the ideal distribution can be upper bounded by
    \begin{equation}
        \TVD(P_\mathrm{real}, P_\mathrm{ideal}) \leq (1-\epsilon n) \TVD(P_\mathrm{real}, P_\mathrm{ideal}) + \epsilon n = (1-\epsilon n)\sqrt{\delta_f} + \epsilon n = \sqrt{\delta_f} + O(1),
    \end{equation}
    where in the last step we use $\epsilon \ll 1/n$ and $\delta_f < 1$.
\end{proof}


\section{Relaxing the assumption that the trials are i.i.d.}
\label{supp:martingale}

Our protocol consists of $N_M$ repeated trials or experiments that are carried out by the prover and the verifier. In the preceding discussion, we have assumed that these trials are independent and identically distributed (i.i.d.), so that the accuracy of our protocol can be shown using simple large-deviation bounds, such as Hoeffding's inequality. Here we sketch how this i.i.d.\ assumption can be relaxed. In this case, the accuracy of our protocol can be shown using large-deviation bounds based on martingales, such as Azuma's inequality \cite{dubhashi2009concentration}. 

To demonstrate this, consider a protocol that estimates the expectation value of an observable $A$ by repeating an experiment (preparing a quantum state and measuring it) $N_M$ times. More complicated protocols can be handled in a similar way. For $j=1,2,\ldots,N_M$, let $F_j$ be the random variable that represents the classical measurement outcome from the $j$th repetition of the experiment. Let $F = (1/N_M) \sum_{j=1}^{N_M} F_j$ be the average of the $F_j$, which we use to estimate the expectation value of $A$. In addition, assume that the operator norm of $A$ is bounded by $\norm{A} \leq \beta$, where $\beta$ is independent of the size of the system, and hence $\abs{F_j} \leq \beta$. This assumption is satisfied for many commonly-used measurements, such as computational-basis measurements preceded by arbitrary single-qubit rotations.

In the case where the trials are i.i.d., the same quantum state $\rho$ is prepared in every trial, and the random variables $F_j$ are i.i.d.\ with expectation value $\Tr(A\rho)$. Then $F$ has expectation value $\Tr(A\rho)$, and Hoeffding's inequality implies that $F$ satisfies a Gaussian-like tail bound with width $O(\beta/\sqrt{N_M})$.

In the non-i.i.d.\ case, it is possible for the $N_M$ trials to be correlated.
Without loss of generality, we can imagine that there exists a joint state $\sigma$ on $N_M$ copies of the quantum system, and for each $j$, the random variable $F_j$ comes from measuring the reduced state on the $j$th copy of the system, which we denote $\sigma_j := \Tr_{\lbrace 1,\ldots,N_M \rbrace \setminus \lbrace j \rbrace} (\sigma)$. 

Despite these complications, it is still possible to interpret $F$ as an estimate of the expectation value of $A$ for a particular quantum state $\tau$ on a single copy of the system. This follows since each $F_j$ has expectation value $\Tr(A\sigma_j)$, and hence $F$ has expectation value $(1/N_M) \sum_{j=1}^{N_M} \Tr(A\sigma_j) = \Tr(A\tau)$, where $\tau = (1/N_M) \sum_{j=1}^{N_M} \sigma_j$. 

Furthermore, despite the fact that the random variables $F_j$ are correlated, one can still show that $F$ satisfies a Gaussian-like tail bound with width $O(\beta/\sqrt{N_M})$. Intuitively, this is because each $F_j$ can influence the value of $F$ by an amount that is bounded by $\pm \beta/N_M$. Formally, this can be shown by well-known martingale techniques (see \cite{dubhashi2009concentration}), i.e., constructing the Doob martingale $G_j = \mathbb{E}(F|F_j,\ldots,F_1)$, showing that $G_j$ has bounded differences $\abs{G_j - G_{j-1}} \leq 2\beta/N_M$, and applying Azuma's inequality.


\section{Echo for more general Hamiltonians}
\label{supp:echo_general}

In the main text, we have shown that the echo approach can be used to generate the single-step history state for a $(ZZ+Z)$-type Hamiltonian on a bipartite lattice. In this section, we show that the single-step history state can be prepared for some---though not all---other Ising-type Hamiltonians.

A $(ZZ+Z)$-type Hamiltonian is very special because its terms commute. This allows us to run the controlled-$Z$s independently and only worry about controlled-$ZZ$s. For more general non-commuting Hamiltonians, we may have to ``invert" all its terms in the echo approach. Under suitable conditions, we can do this using the following theorem.

\begin{theorem}
    If there exists an operator $P$ which is a product of single-qubit operations such that $PHP=-H$, then the single-step history state can be prepared using 2-local operations and controlled-$P$ gates.
\end{theorem}

\begin{proof}
    We start with the initial state $(\ket{0}+\ket{1})\ket{\phi}$ and perform $CP$ before and after a half-time evolution of $H$, followed by a Pauli-$X$ on the clock qubit and a half-time evolution of $H$. The final state is
    \begin{equation}
        \begin{aligned}
            e^{-\imag HT/2}\cdot X_0 \cdot CP\cdot e^{-\imag HT/2} \cdot CP \, (\ket{0}+\ket{1})\ket{\phi}
            &= e^{-\imag HT/2} \left[\ket{1}e^{-\imag HT/2} \ket{\phi} +\ket{0} P e^{-\imag HT/2} P \ket{\phi}\right]\\
            &= e^{-\imag HT/2} \left[\ket{1}e^{-\imag HT/2} \ket{\phi} +\ket{0}  \sum_k \frac{1}{k!} P(-\imag H T/2)^k P \ket{\phi}\right]\\
            &= e^{-\imag HT/2} \left[\ket{1}e^{+\imag HT/2} \ket{\phi} +\ket{0}  \sum_k \frac{1}{k!} (+\imag H T/2)^k \ket{\phi}\right]\\
            &= e^{-\imag HT/2} \left[\ket{1}e^{-\imag HT/2} \ket{\phi} +\ket{0}  e^{+\imag HT/2}  \ket{\phi}\right]\\
            &=\ket{0}\ket{\phi} + \ket{1}e^{-\imag HT}\ket{\phi},
        \end{aligned}
    \end{equation}
which is the desired output.
\end{proof}

There are several cases in which an operator $P$ satisfying the conditions of the theorem can be constructed. For example, if the Hamiltonian consists of $ZZ$ terms on a bipartite interaction graph, then $P$ can apply an $X$ (or $Y$) operator to all qubits on one half of the bipartition. We can also handle some cases where the Hamiltonian is non-commuting, such as
\begin{equation}
    H=\sum_{(i,j)\in\mathrm{NN}} (X_i X_j + Y_i Y_j) + \sum_i Z_i
\end{equation}
acting on a bipartite lattice. Then we can split the system into two sets of qubits where all interactions are between qubits in different sets. If the operator $P$ acts with $X$ on the first set of qubits and $Y$ on the second set, then it anticommutes with each term of $H$, so it has the desired behavior.

\end{document}